\documentclass[aps,pra,reprint,twocolumn,superscriptaddress,nofootinbib]{revtex4-1}
\pdfoutput=1 

\usepackage{graphicx}
\usepackage{bm}
\usepackage{amsmath}
\usepackage{amssymb}
\usepackage{amsfonts}
\usepackage{amsthm}
\usepackage{mathtools}
\usepackage{hyperref}
\usepackage{braket}
\usepackage[normalem]{ulem}
\usepackage{wrapfig}
\usepackage{tikz}
\usepackage{dsfont}
\usepackage{comment}

\hypersetup{colorlinks=true,urlcolor=[rgb]{0,0,0.5},citecolor=[rgb]{0.5,0,0},linkcolor=[rgb]{0,0,0.4}}

\usetikzlibrary{decorations.pathmorphing}

\newtheorem{theorem}{Theorem}
\newtheorem{definition}{Definition}
\newtheorem{corollary}{Corollary}
\newtheorem{lemma}{Lemma}
\newtheorem{proposition}{Proposition}
\newtheorem*{theorem*}{Theorem}
\newtheorem*{proposition*}{Proposition}

\def\autorefapp#1{\hyperref[#1]{Appendix~\ref{#1}}}

\def\ep{\varepsilon}
\def\tr{{\rm tr}}
\def\trs{{\rm tr\,}}
\def\op{{\cal O}}

\def\ketbra#1{ |{#1}\rangle\!\langle{#1}| }
\def\iden{\mathds{1}}
\def\ni{\noindent}

\def\where{\quad {\rm where} \quad}

\def\and{\quad {\rm and} \quad}
\def\ie{{\rm i.e.\ }}

\def\CE{{\cal E}}
\def\CH{{\cal H}}

\def\BE{\mathbb{E}}
\newcommand{\norm}[1]{\left\lVert#1\right\rVert}

\begin{document}
\title{Fluctuations of subsystem entropies at late times}

\author{Jordan Cotler}
\email{jcotler@fas.harvard.edu}
\affiliation{Society of Fellows, Harvard University, Cambridge, MA 02138}

\author{Nicholas Hunter-Jones}
\email{nickrhj@pitp.ca}
\affiliation{Perimeter Institute for Theoretical Physics, Waterloo, ON N2L 2Y5}

\author{Daniel Ranard}
\email{dranard@stanford.edu}
\affiliation{Stanford Institute for Theoretical Physics, Stanford University, Stanford, CA 94305}

\begin{abstract}
We study the fluctuations of subsystem entropies in closed quantum many-body systems after thermalization.  Using a combination of analytics and numerics for both random quantum circuits and Hamiltonian dynamics, we find that the statistics of such entropy fluctuations is drastically different than in the classical setting.  For instance, shortly after a system thermalizes, the probability of entropy fluctuations for a subregion is suppressed in the dimension of the Hilbert space of the complementary subregion.  This suppression becomes increasingly stringent as a function of time, ultimately depending on the exponential of the Hilbert space dimension, until extremely late times when the amount of suppression saturates.  We also use our results to estimate the total number of rare fluctuations at large timescales.  We find that the ``Boltzmann brain'' paradox is largely ameliorated in quantum many-body systems, in contrast with the classical setting.
\end{abstract}

\maketitle

\section{Introduction}
Since the birth of statistical mechanics, there has been intense interest in studying how closed classical systems approach equilibrium, and the statistics of their fluctuations thereafter.  A central difficulty is precisely justifying how deterministic dynamics gives rise to a statistical description.  These same questions and difficulties have been recapitulated in the quantum setting, with significant progress in the last twenty years due in part to the development of quantum information techniques in many-body physics.

We will be interested in the question: what is the time-dependence of a spatial subsystem of a closed, finite-dimensional quantum-many body system after equilibration?  Common lore suggests that the fluctuation statistics of subsystems look essentially the same at any point in time post-equilibration. For instance, if a subsystem's equilibrium entropy is $S$, then it is often said that the probability it dips down to $S'$ is proportional to $\sim \exp(S'-S)$.  Surprisingly, this is not the case in the \textit{quantum} setting.  Previous work~\cite{deutsch1991quantum,srednicki1994chaos,popescu2006foundations, reimann2008foundation,linden2009quantum,d2016quantum} has established that infinite-time-averaged fluctuations of a subsystem are exponentially suppressed in the Hilbert space dimensional of the complement.  While these results are clearly distinct from the classical setting, the degree to which they are an artifact of the infinite-time average was unclear.  In particular, the average is dominated by the behavior of the system at extremely late times, which includes the behavior of Poincar\'{e} recurrences \cite{bocchieri1957quantum,schulman1978note}  and other exotic phenomena which are physically inaccessible.

Through analytic and numerical methods, we will quantitatively demonstrate that a large suppression of fluctuations already occurs shortly after equilibration, and that this suppression becomes more extreme as a function of time.  Interestingly, these findings enable new predictions in experimentally testable regimes.  The suppression of fluctuations is rooted in the development of long-range entanglement, which has no counterpart in classical many-body systems.  Importantly, long-range entanglement continues to be generated long after equilibrium occurs.  Moreover, our findings interpolate between known results about fluctuations at the timescale of equilibration and infinite-time-averaged fluctuations, which capture extremely late-time behavior.  For related work characterizing rare fluctuations, see also \cite{deutsch2020probabilistic,faiez2020typical}.

Specifically, for a quantum many-body system, we consider the von Neumann entropy $S(\rho_A(t))$ of the reduced density matrix of a contiguous subregion $A$ as a function of time, and study its fluctuations. We can get explicit bounds when the evolution is given by a random quantum circuit (RQC), and see that these bounds are parametrically tight in numerically accessible regimes.  Furthermore, we study the statistics of $S(\rho_A(t))$ in systems with Hamiltonian dynamics, and find that these statistics recapitulate the RQC bounds and numerics.

The statistics of post-equilibrium fluctuations is fundamental to the Boltzmann brain paradox.  The debate over Boltzmann brains has a long history (see~\cite{carroll2017boltzmann} for a recent review) and raises confusing questions in light of modern cosmology \cite{dyson2002disturbing, linde2007sinks, bousso2008boltzmann, page2008return, de2010boltzmann, davenport2010there, boddy2017boltzmann, carroll2017boltzmann}. Our results bring a new perspective on the suppression of rare fluctuations by a mechanism unique to quantum many-body systems. We will formulate a precise version of the Boltzmann brain paradox in a simple setting and provide a quantitative analysis towards its resolution.  It will be essential to discuss the nature of fluctuations and measurements \textit{within} a closed quantum system, as in \autoref{sec:interp}.

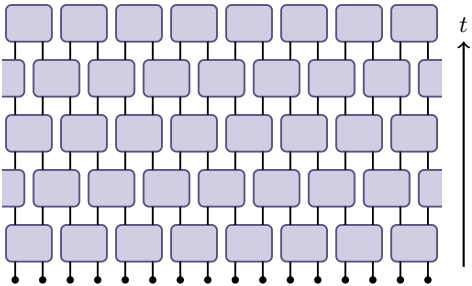
\begin{figure}
\centering
\includegraphics[width=0.68\linewidth]{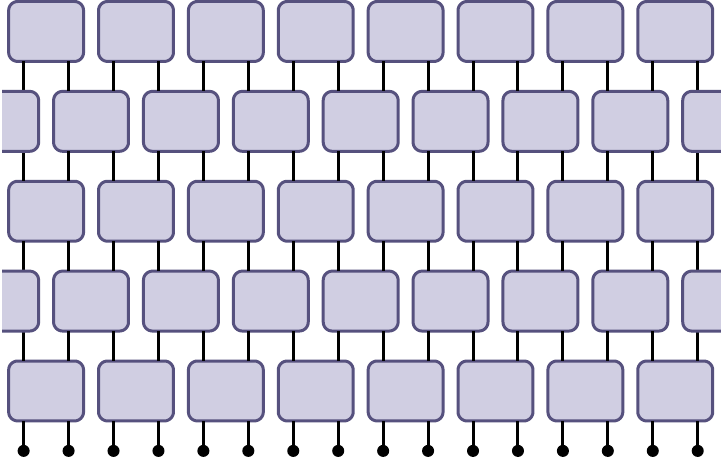}
\begin{tikzpicture}[scale=0.6,baseline=-0.3cm]
\draw[thick,->] (0,0) -- (0,5) node[anchor=south] {{$t$}};
\end{tikzpicture}
\caption{Diagram of a brickwork RQC on a 1D ring of $n$ qudits, evolved to a depth $t$.}
\label{fig:RQC}
\end{figure}

\section{Late-time entropy fluctuations for Random Quantum Circuits} \label{sec:RQC_analytics}
 
RQCs are analytically tractable models of quantum many-body dynamics.  The basic setup is to consider discrete-time dynamics where at each time step spatially local unitaries are applied to the system. The unitaries are each independently and randomly chosen from some common distribution so that the unitary time evolution is Markovian while still respecting spatial locality.  This randomness in both space and time is a key technical tool in the analysis of RQCs.  While RQC time evolution forgoes energy conservation, other dynamical features are know to have quantitative commonalities with Hamiltonian systems, such as the spreading of entanglement and quantum information scrambling, among others \cite{HaydenPreskill,BF12,ODP07,nahum2017entgrowth,nahum2018opgrowth,vonKey2018hydro}. With these features in mind, we initiate our investigation of subsystem entropy fluctuations for random quantum circuits, and later use our results as a point of departure to understand the Hamiltonian setting.

For concreteness, we consider a particular random circuit model, although some of our results will be significantly more general.  Consider a 1D periodic chain of $n$ $q$-level systems which we refer to individually as qudits.  The total dimension of the Hilbert space is $d = q^n$, where we suppose for simplicity that $n$ is even.  The qudits are jointly initialized in a product state; we will see that the particular choice of product state does not matter.  We choose the layout of the gates to form a brickwork pattern, depicted in \autoref{fig:RQC}.  As seen in the figure, at each time step we lay down staggered layers of $n/2$ unitaries independently drawn from the Haar measure on $U(q^2)$ as nearest-neighbor gates. Due to the placement of the random unitaries and the fact that the Haar measure on $U(q^2)$ is $U(q) \times U(q)$--invariant, the probability distribution over states after the first time step (and accordingly thereafter) is independent of the choice of initial product state.  We will refer to this entire RQC setup as the ``1D RQC brickwork model.''

We begin by analyzing the entropy fluctuations of subsystems at pre-equilibrium timescales, to orient our understanding of fluctuations at post-equilibrium timescales.  Let $\rho(t)$ denote the state of the system after $t$ time steps, and accordingly let $\rho_A(t)$ be the reduced density matrix of a contiguous subregion $A$ with associated Hilbert space dimension $d_A$.  It is convenient to denote the complement of $A$ by $B$, and label its Hilbert space dimension by $d_B$.  The maximum von Neumann entropy of $\rho_A(t)$ is $\log (d_A)$, corresponding to the maximally mixed state $\iden_A/d_A$.  In our RQC model, for $A$ being less than half of the system, $\rho_A(t)$ approaches maximal entropy during the process of equilibration, and so we will be interested in deviations from maximally entropy and likewise deviations from the state being maximally mixed.  We will phrase our results in terms of probabilities over the space of admissible RQCs.  Consider the following:

\begin{theorem}[Fluctuation bound at early times]
\label{thm:earlytimebound}
For 1D brickwork RQCs on $n$ qudits with local dimension $q$ and of depth $t$, the entropy of the evolved state $\rho_A(t)$ obeys
\begin{equation}\begin{split}
\label{E:earlyentropybound}
&\Pr \Big(S\big(\rho_A(t)\big)\leq \log(d_A) -\delta \Big)\\
&\qquad\leq \frac{1}{e^\delta-1} \left(\frac{d_A}{d_B} + d_A \left(\frac{2q}{q^2+1}\right)^{2(t-1)}\right)\,,
\end{split}\end{equation}
and the distance to the maximally mixed state obeys
\begin{equation}\begin{split}
\label{E:early1normbound}
&\Pr\Big(\big\|\rho_A(t) - \iden_A/d_A\big\|_1\geq \delta \Big)\\
&\qquad\leq \frac{1}{\delta^2} \left(\frac{d_A}{d_B} + d_A \left(\frac{2q}{q^2+1}\right)^{2(t-1)}\right)\,.
\end{split}\end{equation}
\end{theorem}
The proof of this theorem, as well as all other proofs in this section, can be found in \autorefapp{app:rqcflucs}.  
The theorem tells us that with high probability in a given quantum circuit, the entropy of $\rho_A(t)$ approaches its maximal value at times $t\sim \log(d_A)$, and accordingly $\rho_A(t)$ becomes approximately maximally mixed. Moreover, when $t\sim \log(d_B)$, the right-hand sides of the inequalites in both~\eqref{E:earlyentropybound} and~\eqref{E:early1normbound} are on the order of $\sim e^{-\text{poly}(\delta)} d_A/d_B$.  Thus the timescale for local equilibration is $t\sim \log(d_A)$, in the sense that regions like $A$ undergo only small fluctuations after this time.  On the other hand, the size of typical fluctuations on small regions continues to decrease exponentially until $t = O(n)$, the global equilibration time, which is also the time at which fluctuations on any subextensive region become small.

While \autoref{thm:earlytimebound} holds for all times $t$, it appears to only capture the behavior of rare subsystem fluctuations for times $t \leq O(n)$, \ie before the global equilibration time.  Once we reach times beyond $O(n)$, the size of \textit{typical} fluctuations plateaus, but we we will find that the suppression of \textit{rare} fluctuations becomes continually more extreme.

Now we turn study to post-equilibration time scales, and develop bounds on entropy fluctuations for $t \geq n$.  A key technical tool will be unitary $k$-designs, which are reviewed in the Appendix.  At a high level, an ensemble $\mathcal{E}$ of unitaries (such as an ensemble of RQCs with $t$ layers) is a $k$-design if its first $k$ moments agree with those of the ensemble of Haar-random unitaries with the same dimensions \cite{Dankert09,Gross07}. In our applications, we will consider RQC ensembles which form \textit{approximate} unitary $k$-designs, as given in \autoref{def:approxkdesign} in the Appendix.  Note that it is non-trivial to show that an RQC ensemble forms an approximate design: while RQCs are built out of \textit{local} random unitaries, we are asking that they emulate a \textit{non-local} random unitary which is the size of the entire system.

Nonetheless, it is believed that 1D RQC models form approximate $k$-designs after $t = O(nk)$ layers.  Since the equilibration time for the entire system is $O(n)$, it would take around $k$ equilibration times for the RQC ensemble to form an approximate $k$-design.  As usual, there is a gap between what is believed and what can be proved.  To be precise and rigorous, we first phrase our results in the language of unitary designs, and then discuss what can presently be proven about the relationship between approximate $k$-designs and circuit depth for 1D RQCs. Concentration bounds for unitary designs have been studied before in \cite{LowDeviation09}, where similar bounds were given for the entropy and trace distance. Our results improve on these bounds, but the approach was largely inspired by theirs.
We have the following:
\begin{figure}
\centering
\hspace*{-12pt}
\begin{tikzpicture}[scale=0.88]
\draw[thick,dashed,color=purple] (0.8,0) -- (0.8,3.95);
\node at (1.32,4.2) {\textcolor{purple}{$t\sim \log(d_A)$}};
\draw[thick,dashed,color=purple] (4.5,0) -- (4.5,3.95);
\node at (4.62,4.2) {\textcolor{purple}{$t\sim e^n$}};
\draw[blue, thick,rounded corners] (0,3.4) -- (0.6,3.4) to[out=0, in=136] (4.5,0.4) -- (7.9,0.4);
\draw[thick,->] (0,0) -- (8,0);
\node[rotate=90] at (-0.8,2.2) {Pr(fluctuation)};
\draw[thin] (-0.15,3.4) -- (0,3.4);
\node at (-0.3,3.45) {$1$};
\draw[thin] (-0.15,0.4) -- (0,0.4);
\node at (-0.75,0.45) {$\sim 1/e^{e^n}$};
\draw[thick,->] (0,0) -- (0,4);
\node at (4,-0.5) {Time (circuit depth)};
\end{tikzpicture}
\caption{Schematic of the time dependence (on a log-log scale) of the probability of that the entropy of $\rho_A(t)$ undergoes a fluctuation of some fixed magnitude away from its average value, as bounded by \autoref{cor:latetimes}.}
\label{fig:entflucs}
\end{figure}
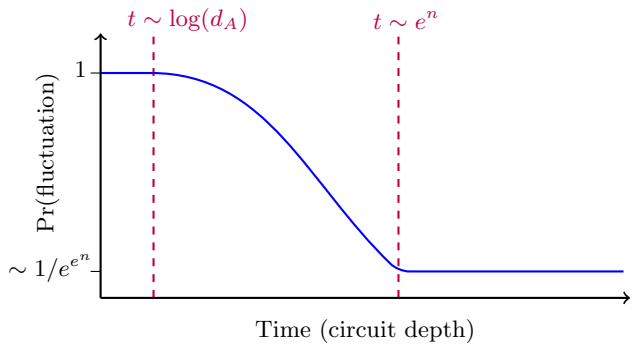

\begin{theorem}[Fluctuations for approximate designs]
\label{thm:appdesignbound}
For an approximate unitary $4k$-design $\CE$, the entropy $S(\rho_A)$ of $\rho_A = \tr_B (U\ketbra{\psi}U^\dagger)$, where $U$ is drawn from $\CE$, obeys
\begin{equation}
\Pr\big( S(\rho_A) \leq \log(d_A) - \delta \big) \leq 2\left(k!+\frac{1}{d^{k}}\right) \left(\frac{9\pi^3}{\gamma^2} \frac{d_A}{d_B}\right)^k,
\end{equation}
where $\gamma:= e^\delta -1 -\frac{d_A}{d_B}$ and for $\delta \geq \frac{d_A}{d_B}$. 
Similarly, the distance between $\rho_A$ and the maximally mixed state $\iden_A/d_A$ obeys
\begin{equation}
\Pr\Big( \big\|\rho_A - \iden_A/d_A\big\|_1 \geq \delta\Big) \leq 2\left(k!+\frac{1}{d^{k}}\right) \left(\frac{9\pi^3}{\eta^2} \frac{d_A}{d_B}\right)^k,
\end{equation}
where $\eta:= \max\{\delta^2,e^{\delta^2/2} -1\} -\frac{d_A}{d_B}$ and taking $\delta^2 > \frac{d_A}{d_B}$.
\end{theorem}
Suppose that $A$ is small compared to the total system size, in which case $d_A/d_B \sim 1/d$.  Then the above theorem tells us that the probability of large entropy fluctuations is suppressed as $\sim (k/d)^k$ and gets exponentially smaller in the design order $k$. Furthermore, for $k \gtrsim d$ the bounds stop improving. 
(As we will see, this suggests that there is a timescale related to $d$ after which entropy fluctuations in RQCs are as small as possible.)

Now we would like to relate the bounds in \autoref{thm:appdesignbound} to circuit depth $t$ of an RQC ensemble.  Although it is expected that we achieve an $\ep$-approximate $k$-design for $t = O(nk)$ for arbitrary local dimension $q$, the closest rigorous statement is:
\begin{theorem}[Linear design growth \cite{NHJ19}]
\label{thm:lineardesigngrowth}
The 1D RQC brickwork model with local dimension $q$ forms an $\ep$-approximate unitary $k$-design when the circuit depth is $t=O(n k)$ for some $q>q_0$, where $q_0$ depends on $k$. 
\end{theorem}
\noindent This result evidently requires large $q$. For $q=2$, i.e.~qubits, one of the best results is that the 1D RQC brickwork model forms an approximate unitary $k$-design for depth $t = O(n k^{11})$ for $k \leq \sqrt{d}$ \cite{BHH12}.

In light of \autoref{thm:lineardesigngrowth} and the fact that we expect our 1D RQCs will become approximate unitary $k$-designs for $t = O(nk)$ even for $q = 2$, we formulate the following corollary to \autoref{thm:appdesignbound}, where we have simplified the form of the inequalities for clarity:
\begin{corollary}[Fluctuation bound at late times]
\label{cor:latetimes}
Suppose there is an RQC which becomes an approximate unitary $k$-design after $t = C n k$ layers for some constant $C$ depending on $q$.  We take $d_A = O(1)$ and let $n$ be large.  Then the entropy fluctuations of $\rho_A(t)$ are suppressed exponentially in the circuit depth $t$ as
\begin{equation}
\!\Pr\big( S(\rho_A(t)) \leq \log(d_A) - \delta \big) \!\lesssim\! \begin{cases}
\left(\frac{1}{e^{2\delta}}\frac{t}{d}\right)^{t/(C n)}, & \!\!\!\! t \leq C' d \\ \\
e^{-C' d/n}\,, & \!\!\!\!t > C' d
\end{cases}
\end{equation}
for a constant $C' < 1$ and $\delta \geq \frac{d_A}{d_B}$, and similarly
\begin{equation}
\!\Pr\big( \|\rho_A(t) - \iden_A/d_A\|_1 \!\geq \delta \big) \!\lesssim \!\begin{cases}
\left(\frac{1}{\delta^4}\frac{t}{d}\right)^{t/(C n)}, & \!\!\!\! t \leq C' d \\ \\
e^{-C' d/n}\,, & \!\!\!\! t > C' d
\end{cases}
\end{equation}
where $C' < 1$ and $\delta^2 > \frac{d_A}{d_B}$.
\end{corollary}
\noindent We note that for large $q$, the above corollary applies to 1D brickwork random circuits due to \autoref{thm:lineardesigngrowth}, but not up to exponential times $t\sim q^n$.

\autoref{cor:latetimes} tells us that the entropy fluctuations of $\rho_A(t)$ decay \textit{super-exponentially} as a function of depth $t$, and that the fluctuations bottom out at size $\sim e^{-d}$ at times $t \sim d$ and thereafter, where we recall that $d = e^{O(n)}$.  The trace distance between $\rho_A(t)$ and the maximally mixed state follows a similar behavior.  A schematic of this behavior is shown in \autoref{fig:entflucs}.

We emphasize that this time dependence is very different than the conventional picture of entropy fluctuations: the size of fluctuations is not statistically similar at all times post-equilibration, but rather the size continues to shrink rapidly until \textit{exponential} times in the number of sites $n$ after which the fluctuations are \textit{double-exponentially} small in $n$.  Numerical evidence in the next section supports that our bounds capture the essential parametric behavior.

Interestingly, fluctuation statistics do not appear to be an artifact of RQC models, and we will argue that the same parametric time dependence depicted in \autoref{fig:entflucs} is also appropriately realized in Hamiltonian systems.  We will take this up in the next section with an appeal to numerical evidence combined with an understanding of previous results in the literature.

We conclude this section with a notable consequence of the above Theorems for RQCs:
\begin{theorem}[Counting fluctuations]
\label{thm:counting}
For 1D brickwork RQCs on $n$ qubits, let $N_A^{\rm ent}(\delta)$ be the number of discrete times $t$ that a contiguous subsystem $A$ satisfies $S(\rho_A(t)) \leq \log(d_A) - \delta$ for times $c_{\rm th} \log(d_A) \leq t \leq e^{c_{\rm rec} d}$, where $c_{\rm th} > 1$ and $c_{\rm rec} < 1$.  Then for $n$ large enough and the constant $c_{\rm rec} = \gamma^2/(9 \pi^3 d_A^2 e)$, the probability of a single entropy fluctuation is bounded as
\begin{equation}
\Pr\left( N^{\rm ent}_A(\delta)>0 \right) \leq \frac{8}{e^{\delta}- 1 }\,\left(\frac{1}{d_A}\right)^{\frac{2}{5}c_{\rm th} -1}\,.
\end{equation} Similarly, if $N_A^{\rm dist}(\delta)$ is the number of discrete times $t$ that $A$ satisfies $\|\rho_A(t) - \iden_A/d_A\|_1 \geq \delta$ for $t$ in the same range, then for $n$ large enough
\begin{equation}
\Pr\left( N^{\rm dist}_A(\delta)>0 \right) \leq \frac{8}{\delta^2}\,\left(\frac{1}{d_A}\right)^{\frac{2}{5}c_{\rm th} -1}\,.
\end{equation}
\end{theorem}
A proof can be found in \autorefapp{app:cor2proof}. We state the result for local qubits ($q=2$) and note that the bound only improves for larger local dimension.

This result establishes that we do not expect the subsystem $A$ to fluctuate by an $O(1)$ amount \textit{even once} between the vastly difference timescales $c_{\rm th} \log(d_A)$ and $e^{c_{\rm rec} d}$.  Here, $c_{\rm th} \log(d_A)$ is a thermalization timescale for $A$, whereas the $e^{c_{\rm rec} d}$ is parametrically the timescale at which the entire system undergoes a Poincar\'{e} recurrence.  We will discuss implications of this result in \autoref{sec:BB}.

As discussed in \autorefapp{app:therm}, many physical systems may have a polynomially slower approach to global equilibrium.  However, even in the case of slower hydrodynamic relaxation, the charge-conserving RQCs discussed in \autorefapp{app:ccRQC} appear to exhibit $1/d$ suppression after polynomial times.  Moreover, we expect that \autoref{thm:counting} generalizes to the setting where global equilibration occurs on $\text{poly}(n)$ timescales, followed by exponential suppression and ultimately double exponential suppression thereafter.

\section{Hamiltonian evolution and numerics}\label{sec:numerics} 

\begin{figure}
\centering
\begin{tikzpicture}[scale=0.8,baseline=0mm]
\node at (0,0) {\includegraphics[width=0.9\linewidth]{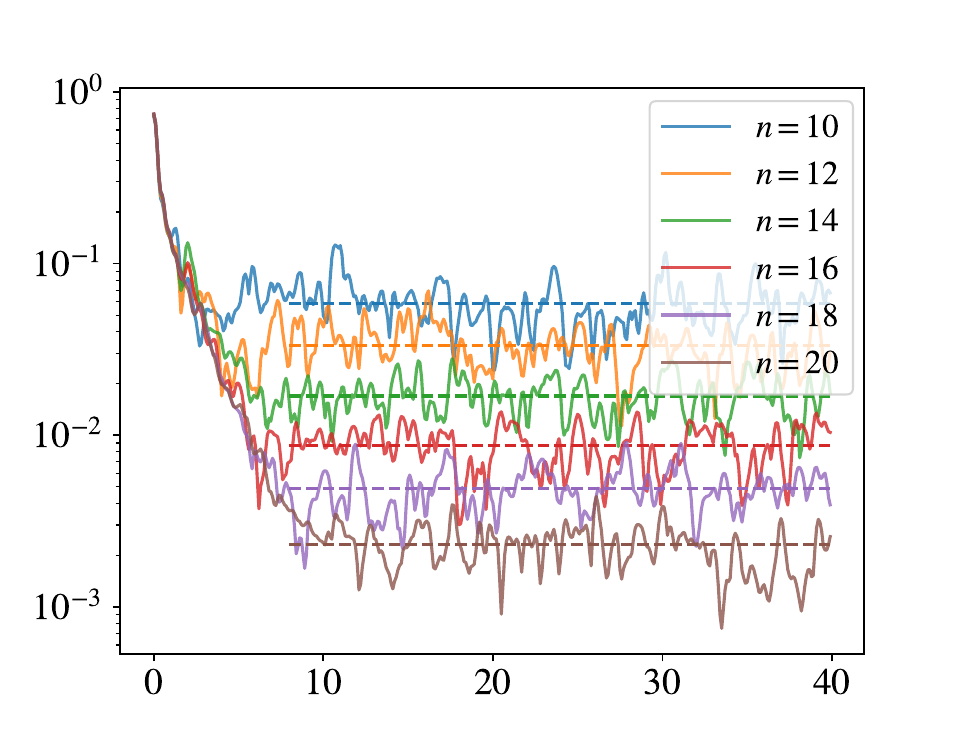}};
\node at (0.1,-3.7) {Time};
\node[rotate=90] at (-4.9,0) {$\|\rho_A(t) - \rho_A^{\rm avg}\|_1$};
\node at (0.1,3.4) {{\footnotesize\bf Relaxation of subsystems for spin chains}};
\end{tikzpicture}
\caption{A nonintegrable spin-$\frac{1}{2}$ chain is simulated numerically, initiated in a fixed, translation-invariant product state. For a subsystem $A$ of two contiguous qubits, the fluctuation $\norm{\rho_A(t)-\rho_A^{\textrm{avg}}}_1$ of $\rho_A(t)$ from its time-averaged value $\rho_A^{\textrm{avg}}$ is plotted over time (log-linear), with different curves for spin chains of varying length $n$.  The dashed line indicates the time-averaged value of the fluctuation, over the duration of the dashed line. The fluctuations decrease exponentially for a time proportional to $n$, plateauing at a value exponentially small in system size.
}
\label{fig:Ham_relax}
\end{figure}
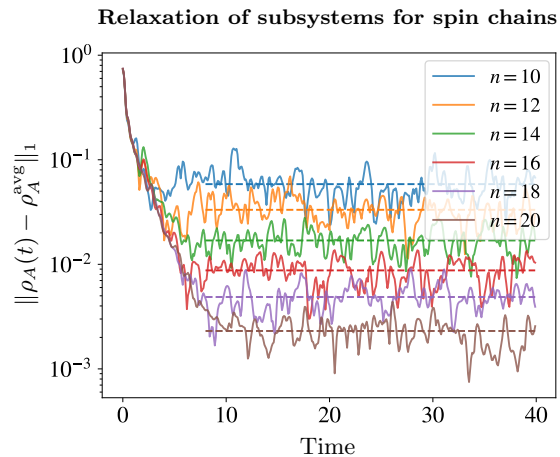

We want to study the size and rarity of fluctuations in systems undergoing chaotic Hamiltonian evolution, drawing insights from our results for RQCs. When discussing RQCs, we characterized the probability of fluctuations with respect to the ensemble of random circuits.  To draw an analogy with Hamiltonian evolution, we first need to ask what it means for a fluctuation to have small probability. 
One natural analog of the RQC ensemble is an ensemble of chaotic Hamiltonians, or alternatively an ensemble of initial states. We focus instead on the frequency of fluctuations over time for a single local Hamiltonian, chosen to be chaotic. During later time windows, we expect that typical fluctuations are smaller and larger fluctuations are more rare. 

For states evolving under chaotic Hamiltonians, the statistics of fluctuations are difficult to study analytically.  Previous analytic results primarily characterize infinite-time averages \cite{reimann2008foundation,farrelly2017thermalization, linden2009quantum, linden2010speed} or averages over finite times that are exponential in system size \cite{short2012quantum}.\footnote{The work \cite{short2012quantum} characterizes average fluctuations over some finite timescale $T$, with conditions on the minimum size of $T$ for the bound to become effective.  For generic chaotic Hamiltonians, this minimal value of $T$ is exponential in system size.} Some results characterize equilibration by examining certain time-dependent observables, whereas others analyze equilibration by leveraging certain assumptions on the Hamiltonian energy spectrum~\cite{garcia2017equilibration, reimann2019transportless,malabarba2014quantum}.  For further review of relevant literature on the relaxation to equilibrium  in many-body systems, including results about infinite-time averages, see \autorefapp{app:therm}.

While fluctuations at infinite timescales are well understood,  we are also interested in (comparatively) earlier times, namely after equilibrium but before exponential times.  In the case of RQCs, \autoref{thm:earlytimebound} and \autoref{cor:latetimes} characterize the increasingly strict suppression of fluctuations throughout this time period.  Here we ask whether chaotic Hamiltonian systems exhibit a similarly drastic suppression.

We begin by considering the decay of fluctuations during the relaxation to equilibrium.  In \autoref{thm:earlytimebound}, we characterized this decay for 1D RQCs of length $n$.  There, the reduced state $\rho_A(t)$ of a small subsystem $A$ approaches its long-time average after $\sim \log(d_A)$ time, and afterwards the size of typical fluctuations decreases exponentially until $O(n)$ time, when the typical distance of $\rho_A(t)$ from its average is at most $1/d$.  

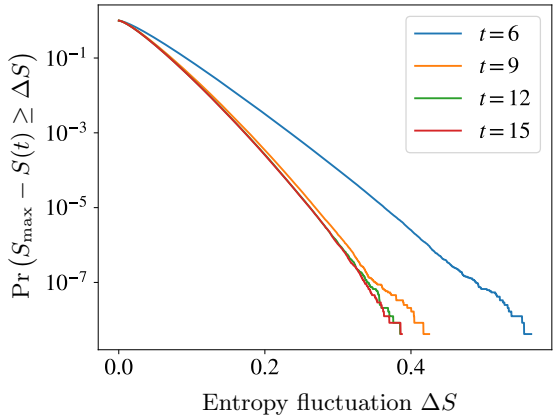
\begin{figure}
\centering
\begin{tikzpicture}[scale=0.8,baseline=0mm]
\node at (0,0) {\includegraphics[width=0.9\linewidth]{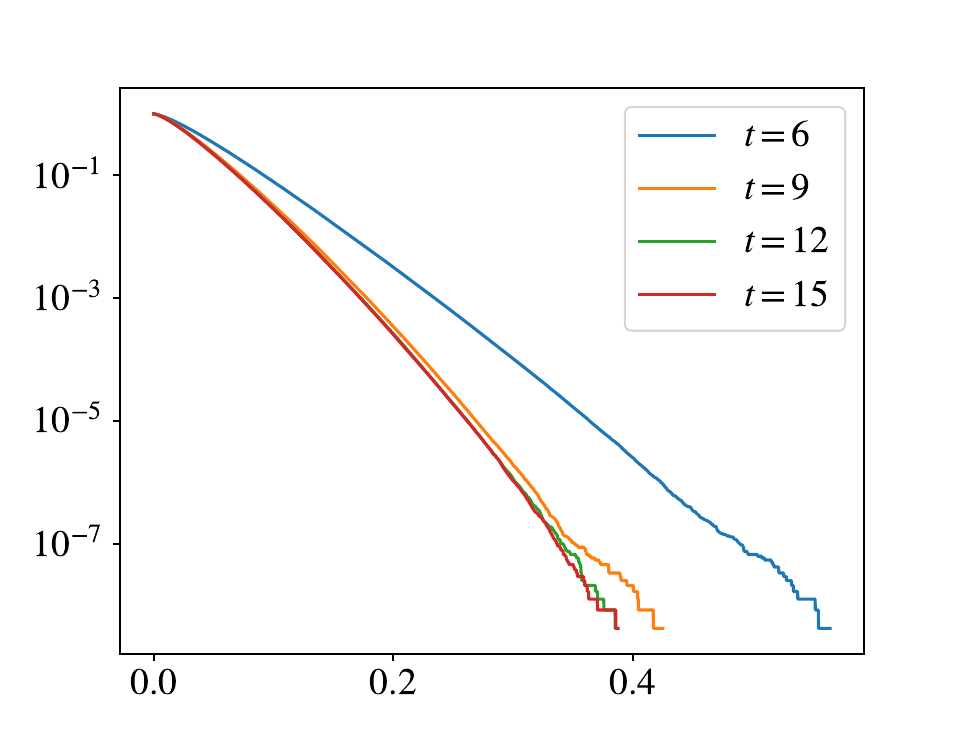}};
\node at (0.2,-3.8) {Entropy fluctuation $\Delta S$};
\node[rotate=90] at (-4.9,0) {$\Pr\big(S_{\rm max} - S(t) \geq \Delta S\big)$};
\node at (0.1,3.4) {{\footnotesize\bf Distribution of entropy fluctuations for RQCs}};
\end{tikzpicture}
\caption{Brickwork RQCs on periodic 1D chains of length 6 are simulated for 15 layers, for $2.5 \times 10^8$ random trials.  Considering a subsystem of a single qubit, we plot the empirical probability ($y$-axis) of entropy fluctuations of size $\Delta S$ ($x$-axis).  Curves are plotted for different times $t$, i.e.\ circuit depth.  We see that the extreme tails of the distribution (lower right) take longer to reach their late-time values.
}
\label{fig:RQC_hist}
\end{figure}

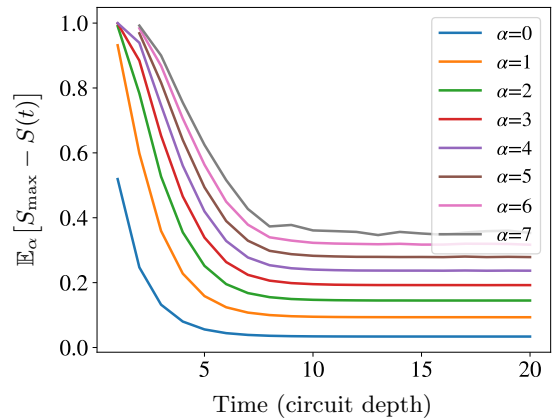
\begin{figure}
\centering
\begin{tikzpicture}[scale=0.8,baseline=0mm]
\node at (0,0) {\includegraphics[width=0.9\linewidth]{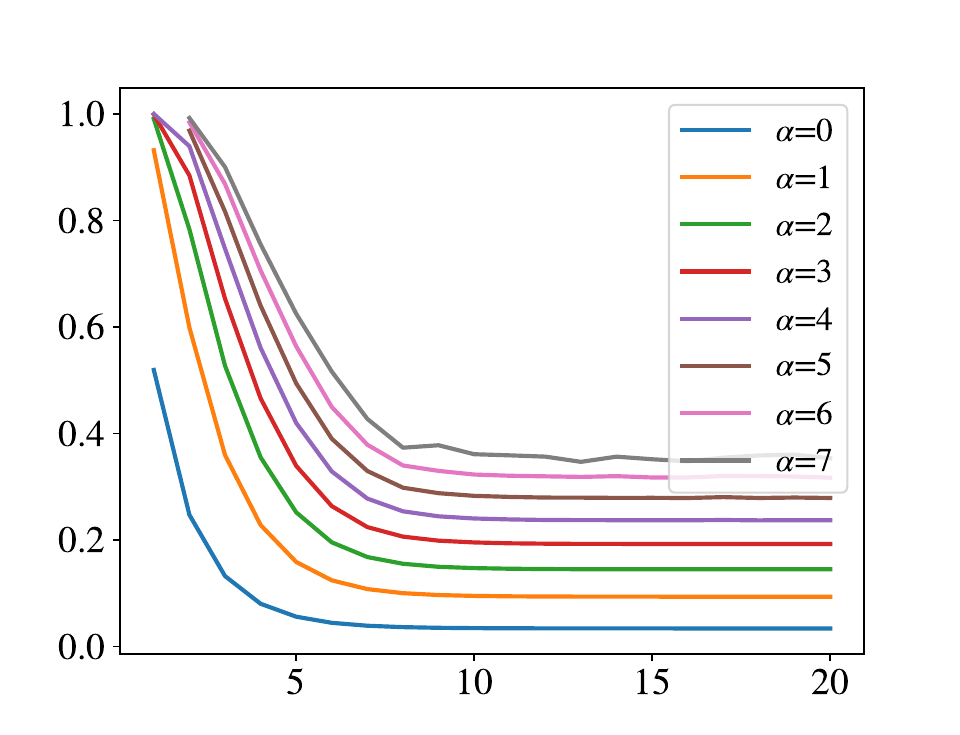}};
\node at (0.1,-3.8) {Time (circuit depth)};
\node[rotate=90] at (-4.8,0) {$\BE_\alpha\big[ S_{\rm max} - S(t)\big]$};
\node at (0.1,3.4) {\vbox{\footnotesize\bf Change in tails of distribution of entropy\\ fluctuations for RQCs}};
\end{tikzpicture}
\caption{Brickwork RQCs on periodic 1D chains of length 6 are simulated for 20 layers, for $2.5 \times 10^8$ random trials.  Focusing on a subsystem of a single qubit, we consider the empirical probability distribution of entropies $S(t)$.  For various values of the parameter $\alpha$, we plot the mean value of the entropy fluctuation (with respect to random trials) when restricted to the rarest $10^{-\alpha}$ fraction of fluctuations, denoted $\BE_\alpha [S_{\textrm{max}} - S(t)]$.  Thus curves at larger $\alpha$ quantify the size of increasingly rare entropy fluctuations.  At larger $\alpha$, the curves take longer to plateau,  demonstrating that the extreme tails of the distribution take longer to reach their late-time values.
}
\label{fig:RQC_tails}
\end{figure}

Here we numerically investigate the analogous regime for spin chains with chaotic Hamiltonians.  Many relevant aspects of such evolution are already well-characterized \cite{reimann2016typical}, including generic thermalization of $O(1)$-sized subsystems within $O(1)$ timescales. We ask whether there is also continued suppression of fluctuations until at least $O(n)$ time, in analogy with \autoref{thm:earlytimebound}.

In particular, we simulate a nonintegrable Hamiltonian on a small spin chain of up to 20 qubits, initiated in a translation-invariant product state.  We use the nonintegrable model studied in \cite{banuls2011strong}, a chaotic Ising model with both transverse and parallel fields, as described in their Eq.\ (1), and we use the same couplings.\footnote{When not otherwise specified, the initial state is the translation-invariant product state $|Y+\rangle^{\otimes n}$, described by \cite{banuls2011strong} as a strongly thermalizing initial condition.  At lower energy densities, the model may have an approximately integrable quasiparticle description \cite{cheng2017quasiparticle}, which is why we primarily use the former state.}  Focusing on a small subsystem $A$ of two adjacent qubits, \autoref{fig:Ham_relax} shows the fluctuation $F(t) = \norm{\rho_A(t)-\rho_A^{\textrm{avg}}}_1$ over time, where $\rho_A^{\textrm{avg}}$ is the long-time average.  We confirm that while the subsystem equilibrates in $O(1)$ time, in the sense that $F(t)$ becomes small, the typical size of fluctuations continues to decay exponentially until $O(n)$, as expected by analogy with \autoref{thm:earlytimebound}. 

When analyzing late-time fluctuation statistics, the analogy between RQCs and Hamiltonian evolution must be made at times larger than any pre-thermalization timescale, and after the relaxation of hydrodynamic modes, as discussed in \autorefapp{app:therm}. After this regime, we might expect that Hamiltonian dynamics ultimately yield an exponential relaxation of fluctuations, until the typical size of fluctuations is $1/d$. This behavior is explored for charge-conserving RQCs in \autorefapp{app:ccRQC}. We leave further investigation to future work.

At later times, \autoref{cor:latetimes} for RQCs suggests that the probability of $O(1)$-sized fluctuations continues to exponentially decay, until exponential times $t \sim d$ when the final suppression is doubly exponential.  To illustrate the full range of this behavior with numerics is extremely difficult:  even for modest system sizes, doubly exponentially suppressed fluctuations will simply not occur when sampling a tractable number of time points.  
However, we can attempt to demonstrate an essential aspect of the above behavior for RQCs: namely, that the \textit{tail} of the distribution of fluctuations is increasingly suppressed in time, even after the \textit{typical} size of fluctuations plateaus.  This phenomenon is captured for short times post-equilibration in the numerics of \autoref{fig:RQC_hist} and \autoref{fig:RQC_tails}, as explained in the figure captions.

Meanwhile, to numerically verify the behavior of \autoref{cor:latetimes} for the case of Hamiltonian evolution is even more difficult than in the case of RQCs.  Nonetheless, we corroborate an aspect of the Hamiltonian behavior.  In \autoref{fig:Ham_hist}, for a small spin chain, we consider evolution until a time $t \sim 10^6$, using the same nonintegrable Hamiltonian as for \autoref{fig:Ham_relax}. The figure depicts the fraction of the time within, for $0 \leq t \lesssim 10^6$, that the entropy deviates by at least $\Delta S$ from its average value. The probability of fluctuations appear exponentially suppressed as $e^{-\Delta S \cdot d}$ within this regime.  The $d$ dependence is demonstrated by the exponential increase in the slope of the curves with respect to $n$ (shown in the inset). This resembles\footnote{Actually, \autoref{cor:latetimes} suggests that fluctuations are further suppressed as  $e^{-e^{\Delta S} d}$. But for the small subsystem and range of $\Delta S$ probed by  \autoref{fig:Ham_hist}, the suppression only appears as $e^{-\Delta S d}$.} the exponential concentration of the infinite-time result in Eq.\ \eqref{eq:inf_time_ham_conc}, but witnessed here already at finite (exponential) times.

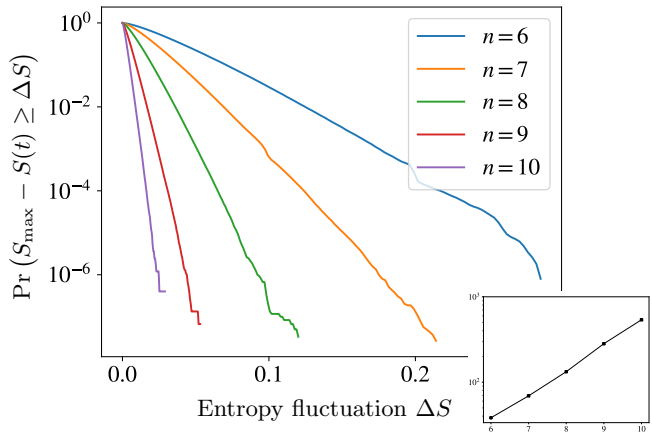
\begin{figure}
\centering
\hspace*{-12pt}
\begin{tikzpicture}[scale=0.8,baseline=0mm]
\node at (0,0) {\includegraphics[width=0.8\linewidth]{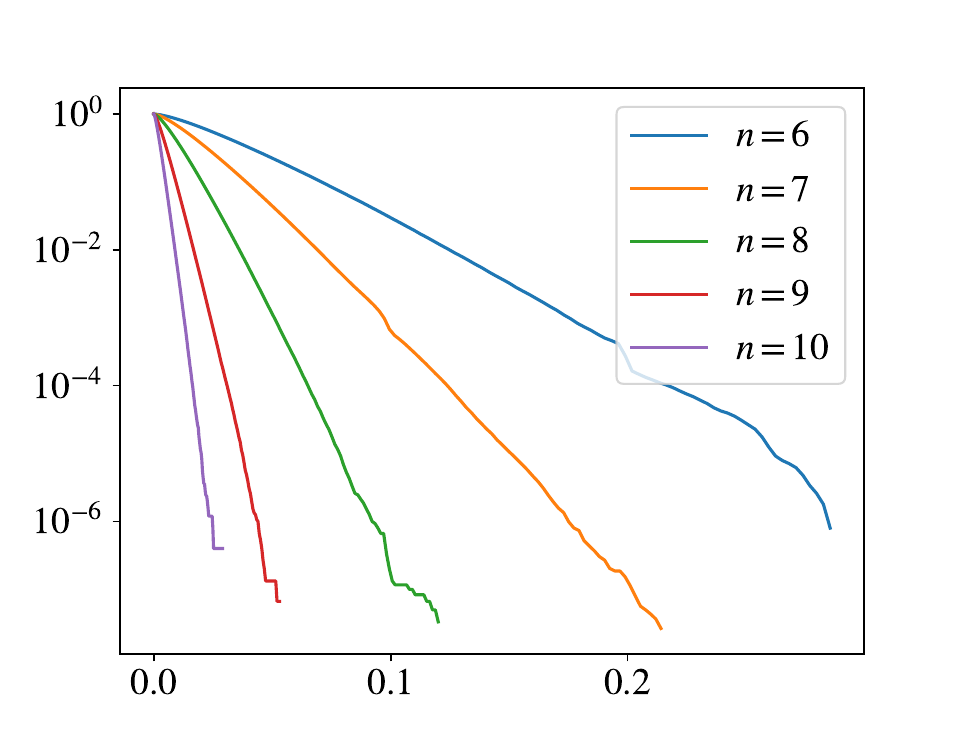}};
\node at (0.3,-3.6) {Entropy fluctuation $\Delta S$};
\node[rotate=90] at (-4.7,0.3) {$\Pr\big(S_{\rm max} - S(t) \geq \Delta S\big)$};
\node at (0.4,3.7) {\vbox{\footnotesize\bf Distribution of post-equilibrium fluctuations before exponential time, for Hamiltonian evolution}};
\node at (4.2,-2.8) {\includegraphics[width=0.28\linewidth]{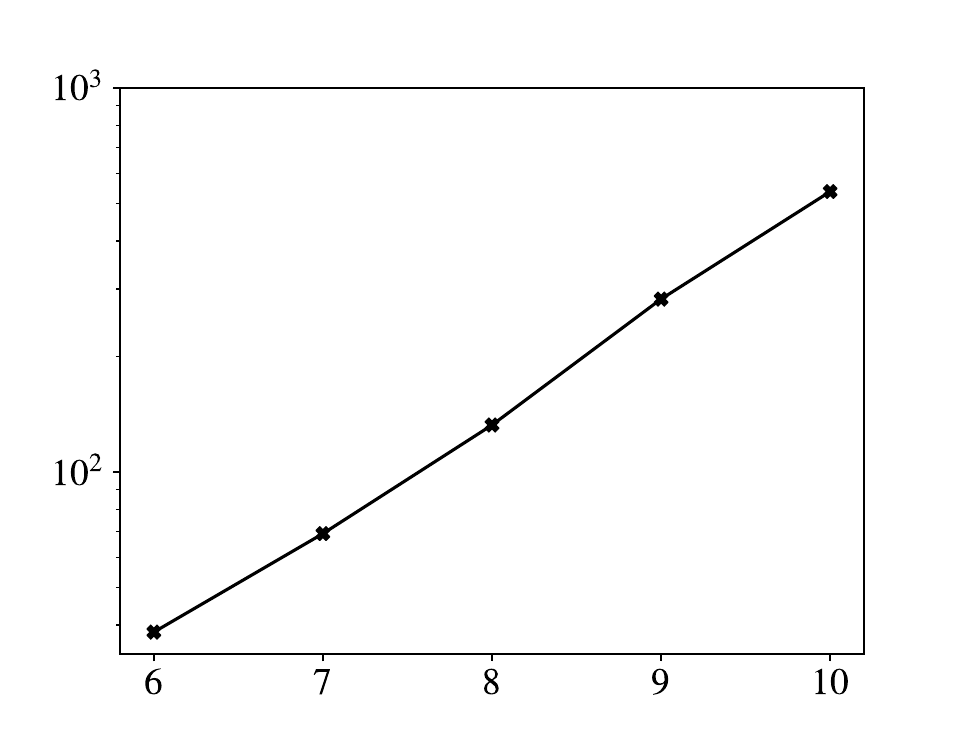}};
\end{tikzpicture}
\caption{A nonintegrable spin-$\frac{1}{2}$ chain is simulated numerically, initiated in a fixed, translation-invariant product state, and evolved until $t \sim 10^6$. For a subsystem $A$ of one qubit, we plot $\textrm{Pr}_t( S_{\textrm{max}} - S(t) > \Delta S )$, the fraction of the time that the entropy $S(t)$ of $\rho_A$ deviates by at least $\Delta S$ from its maximum value of 1.  The plot is log-linear, with different curves for spin chains of varying length $n$.  \textit{Inset}: The best-fit slopes ($y$-axis) of the former curves are plotted with respect to $n$ ($x$-axis) on a log-linear scale, demonstrating an exponential increase with respect to $n$. 
}
\label{fig:Ham_hist}
\end{figure}

Given the increasing suppression of fluctuations at later times suggested by \autoref{cor:latetimes}, one might have the impression that the rarer fluctuations in \autoref{fig:Ham_hist} occurred at earlier times. \autoref{fig:Ham_ttf} indicates that is not the case.  Instead, fluctuations are sufficiently suppressed that for a single evolution, the time required to encounter some $\Delta S$ fluctuation goes like $\sim e^{d \Delta S}$, with no large fluctuations at early times.  This behavior corroborates \autoref{thm:counting} which is for RQCs.

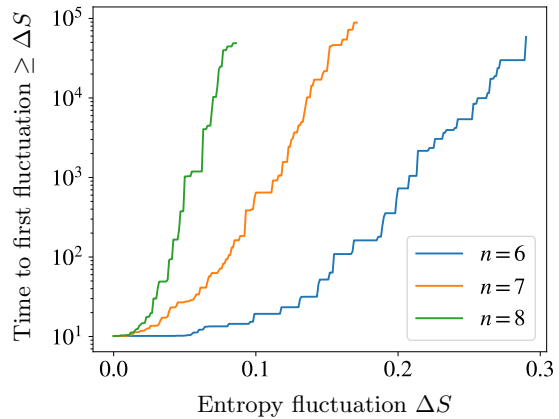
\begin{figure}
\centering
\begin{tikzpicture}[scale=0.8,baseline=0mm]
\node at (0,0) {\includegraphics[width=0.9\linewidth]{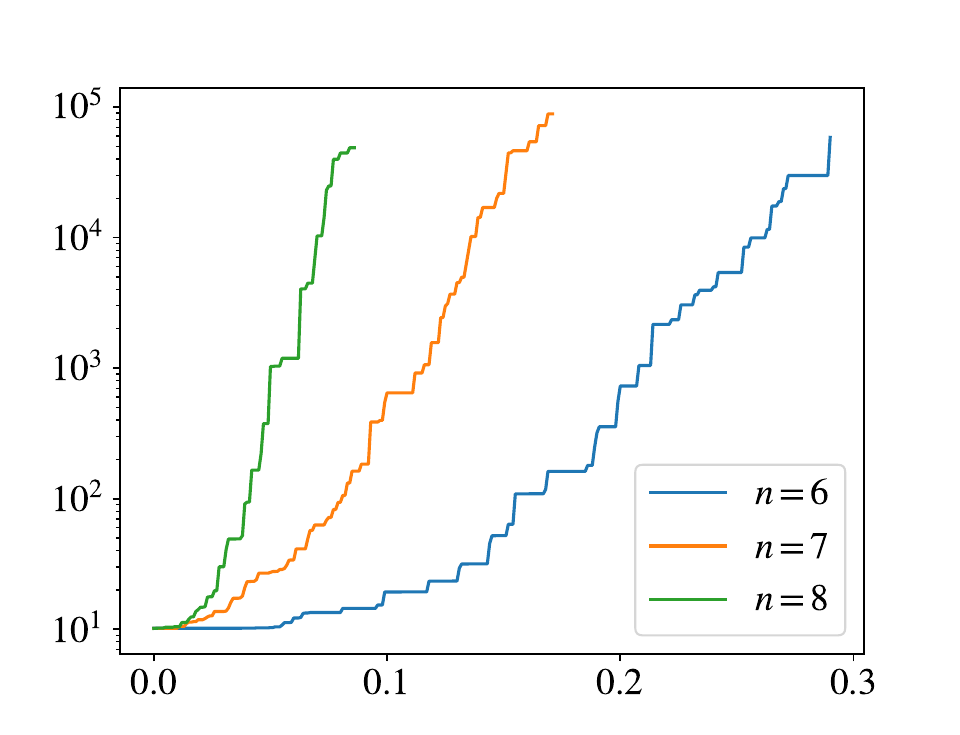}};
\node at (0.2,-3.8) {Entropy fluctuation $\Delta S$};
\node[rotate=90] at (-4.8,0) {Time to first fluctuation $\geq\Delta S$};
\node at (-0.2,3.4) {{\footnotesize\bf Time until fluctuations, for Hamiltonian evolution}};
\end{tikzpicture}
\caption{A nonintegrable spin-$\frac{1}{2}$ chain is simulated numerically, initiated in a product state, and evolved until $t \sim 10^6$. For a subsystem $A$ of one qubit, we ask how long ($y$-axis) until the entropy of the subsystem fluctuates a certain distance $\Delta S$ ($x$-axis) from its maximum value.  Distinct curves show data for spin chains of varying length $n$.  Each curve is averaged over 100 initial states, chosen as random product states.
}
\label{fig:Ham_ttf}
\end{figure}

\section{Boltzmann brains and their suppression}
\label{sec:BB}

Classically, rare fluctuations have been studied in the context of the ``Boltzmann brain'' paradox (see~\cite{dyson2002disturbing, linde2007sinks, bousso2008boltzmann, page2008return, de2010boltzmann, davenport2010there, boddy2017boltzmann, carroll2017boltzmann} for a modern treatment).  Our results quantify the extreme rarity of these fluctuations in quantum systems relative to classical systems.  For a discussion of subtleties about the interpretation of fluctuations in closed quantum systems, see~\autoref{sec:interp} below.

Let us cast the Boltzmann brain problem into a precise mathematical question, in the context of a closed quantum many-body system. The resulting problem is directly related to \autoref{thm:counting} above.  Suppose we consider a closed system with $n$ qubits (this partitioning into qubits is not essential for what follows).  If the system evolves with respect to spatially local dynamics, a subsystem $A$ of size $\log(d_A)$ will thermalize in a time $\sim c_{\rm th} \log(d_A)$ for some $O(1)$ constant $c_{\rm th}$.  Let $S_{\text{init}}$ be the initial entropy of $\rho_A$, and let $S_{\text{eq}}$ be the equilibrium (thermalized) value.  We suppose that $S_{\text{init}} < S_{\text{eq}}$\,, i.e.~that the subsystem is initialized in some comparatively low-entropy state.  The entire system will undergo a Poincar\'{e} recurrence on a timescale $\sim e^{O(d)}$, where $\rho_A$ (and more broadly the entire global state) will return to its initial configuration.  We can ask: if $N_A^{\text{ent}}(\delta)$ is the number of times $S(\rho_A(t))$ fluctuates below $S_{\text{eq}} - \delta$ for $t$ between $c_{\rm th} \log(d_A)$ and $e^{c_{\rm rec} d}$, then is $N_A^{\text{ent}}(\delta) \ll 1$ for $\delta = O(1)$?

In words, the question is asking how many times $A$ fluctuates into a significantly lower entropy configuration after thermalizing (but parametrically before a Poincar\'{e} recurrence).  In the classical setting, $N_A^{\text{ent}}(\delta)$ is easy to estimate.  Classical statistical mechanics dictates that the probability that $A$ fluctuates $\delta$ below its equilibrium value is $\sim e^{-\delta}$, and so we expect $N_A^{\text{ent}}(\delta) \sim (e^{c_{\rm rec} d} - c_{\rm th} \log(d_A)) \, e^{-\delta}$ which is an \textit{enormous} number.  In such a classical universe, \textit{reductio ad absurdum}, the vast majority of humans (or more minimally, disembodied brains referred to as ``Boltzmann brains'') arise as thermal fluctuations away from equilibrium.

In the quantum setting, this is not so.  Assuming the parallel between Hamiltonian evolution and RQCs, our results (in particular, \autoref{thm:counting}) imply that after a closed quantum system \textit{globally} thermalizes (i.e.\ after times of $O(n)$), the total number of fluctuations on a subsystem $A$ away from its equilibrium entropy by an amount $\delta$ is enormously suppressed.  Further, after $A$ equilibrates in a time $c_{\rm th} \log(d_A)$ but before global equilibration, the entropy fluctuations are suppressed by $\sim d_A^{-O(c_{\rm th})}$.  For even modest $c_{\rm th} > 1$, the probability of even one fluctuation by more than $\delta$, $\text{Prob}(N_A^{\rm ent}(\delta)>0)\sim d_A^{-O(c_{\rm th})} \ll 1$, is small.  As such, post-equilibration Boltzmann brains are exceedingly unlikely in closed quantum systems before timescales parametrically like the Poincar\'{e} recurrence.

\section{Interpretation of fluctuations}
\label{sec:interp}

We have referred to changes in a subsystem's reduced state as fluctuations.  However, even when a state is static, one sometimes refers to ``fluctuations'' in the state itself.  For instance, when the reduced state $\rho_A \sim \frac{1}{Z}e^{-\beta H_A}$ on subsystem $A$ is thermal,  one often says it exhibits thermal fluctuations.   

Whether or not thermal fluctuations constitute `true' fluctuations may appear a matter of terminology, but the question is especially relevant to the discussion of Boltzmann brains in \autoref{sec:BB}.  In that context, even a static subsystem is sometimes claimed to fluctuate into a Boltzmann brain state.  We argue that despite this common claim, when a closed system evolves unitarily, subsystems are correctly viewed as static whenever the state is globally equilibrated.\footnote{As an alternative argument, note that such fluctuations entail some semi-classical description of the state to begin with; however, following the framework of \cite{riedel2017classical}, globally equilibrated states would not have a semi-classical decomposition.} A related discussion and overview of opposing arguments may be found in \cite{boddy2017boltzmann}.

In the simplest case, suppose that $\rho$ is the maximally mixed state $\rho = \mathds{1}/d$ of a closed system.  Then $\rho$ is clearly invariant under unitary evolution.  It is sometimes said that the maximally mixed state can fluctuate, since any measurement has a nonzero probability for every outcome.  However, the term ``fluctuation'' is misleading here; $\rho$ does not have any dynamical description. In practice, measurements are performed by coupling $\rho$ to an external measurement apparatus.  Only then does the joint system fluctuate, encoded by the joint correlation.

Now consider a subsystem $A$ of a large closed system.  In some circumstances, a distinct subsystem $B$ might behave like a measurement apparatus within the system, producing a fluctuation by coupling to $A$ and generating correlations in $\rho_{AB}$.  However, these dynamics do not occur in a globally equilibrated state, such as that of our RQC at late times.  There, \textit{every} subextensively sized subsystem is static, up to minuscule corrections.  No subextensive subsystem $B$ couples to $A$ so as to measure it, because $\rho_{AB}$ is also static.  Note this argument holds whether the subsystems involved are maximally mixed, thermal at finite temperature, or merely static.

We conclude that fluctuations only occur when the reduced states of subsystems substantially deviate from their equilibrium values.

\section{Discussion}
We have argued, using a combination of analytic and numerical evidence, that entropy fluctuations in equilibrated quantum many-body systems are suppressed in the total Hilbert space dimension, tending towards \textit{exponential} suppression in the Hilbert space dimension as time progresses.  Furthermore, this behavior is vastly different from that of classical many-body systems, for which more standard thermodynamic intuitions apply.  A consequence of our findings is that in closed quantum many-body systems, the suppression of fluctuations is observable even on shorter timescales after local equilibration, and for small subsystems (say even on 1 or 2 sites).  As such, existing experimental quantum many-body simulators designed for probing dynamics away from equilibrium~\cite{cheneau2012light, meinert2014observation, schreiber2015observation, kaufman2016quantum, lukin2019probing} should be able to measure this suppression, and perhaps probe novel fluctuation statistics which go beyond the abilities of numerics on classical computers.

Further, our results imply that Boltzmann brain fluctuations are only likely to occur at times exponential in the Hilbert space dimension, parametrically on the same time scale as Poincar\'{e} recurrences.  This alleviates, or at least postpones, the Boltzmann brain paradox by a mechanism unique to quantum many-body systems.

From a theoretical point of view, our findings motivate further investigations of post-equilibrium physics in RQCs and Hamiltonian dynamics.  It would be interesting to better understand the role of conserved charges in either setting, beyond the investigation initiated in \autorefapp{app:ccRQC}.  More broadly, it would be desirable to have analytic results in the Hamiltonian setting, even if only for restricted classes of interactions and timescales.

\vspace*{6pt}
\noindent {\bf Acknowledgments.}\quad 
We thank {\'A}lvaro Alhambra, Adam Brown, Curt von Keyserlingk, Micha{\l} Oszmaniec, Tibor Rakovszky, and Jess Riedel for enlightening discussions and comments. 
JC is supported by a Junior Fellowship from the Harvard Society of Fellows, as well as in part by
the Department of Energy under grant {DE}-{SC0007870}.
Research at Perimeter Institute is supported by the Government of Canada through the Department of Innovation, Science and Industry Canada and by the Province of Ontario through the Ministry of Colleges and Universities.

\onecolumngrid
\appendix

\section{Subsystem fluctuations for Haar random states}

Consider a bipartite system $\CH = \CH_A\otimes \CH_B$, with total dimension $d = d_A d_B$ and assume $d_A \leq d_B$. For the random state $U\ket\psi$, where $U$ is drawn Haar randomly from $U(d)$ and $\ket\psi$ is an arbitrary initial state, consider the state $\rho_A = \tr_B U\ketbra{\psi}U^\dagger$. The averaged von Neumann entropy is \cite{Page93} 
\begin{equation}
\BE_U [S(\rho_A)] = \log(d_A) - \frac{d_A}{2d_B} +\frac{1}{2d} + O(d_B^{-2})\,.
\end{equation}
An exact expression for the entropy was also conjectured by Page in \cite{Page93} and later proven.

A natural question is given an instance of a random state, how likely is the entropy of a subsystem to be close to the average. Equivalently, how rare are large fluctuations in the subsystem entropy of a Haar random state. This question was addressed by Hayden, Leung, and Winter in \cite{HaydenGeneric06} where they derived bounds on the deviations from the expected value of $S(\rho_A)$ for Haar random states. Levy's lemma gives that for $d_B\geq d_A\geq 3$
\begin{equation}
\Pr\Big( S(\rho_A) \leq \log(d_A) - \frac{d_A}{d_B} - \delta\Big) \leq \exp\left( - \frac{(d-1)\delta^2}{8\pi^2 (\log(d_A))^2}\right)\,,
\end{equation}
meaning the probability of a deviation away from the average subsystem entropy is suppressed exponentially in the total Hilbert space dimension $d$.

\section{Subsystem fluctuations in random quantum circuits}
\label{app:rqcflucs}

We are interested in bounding the probability that the entropy of a subsystem of a state evolved by random quantum circuit deviates from its maximal value. We again consider a bipartite system $\CH = \CH_A\otimes \CH_B$ with dimension $d = d_A d_B$ and $d_A \leq d_B$. Let us evolve an initial state $\ket\psi$ by a unitary $U_t$, drawn randomly from the set of depth $t$ random circuits, and denote the evolved state as $\rho(t) = U_t \ketbra{\psi} U_t^\dagger$. For the reduced state on subsystem $A$, $\rho_A(t) = \tr_B U_t \ketbra{\psi} U_t^\dagger $, we would like bounds of the form 
\begin{equation}
\Pr\big(S(\rho_A(t))\leq \log(d_A) -\delta\big) \leq f(t)\,,
\end{equation}
where $S(\rho) = -\tr(\rho \log \rho)$ is the von Neumann entropy. If $f(t)$ is substantially suppressed, then with very high probability the subsystem entropy is almost maximal. 

Similarly, we are also interested in the distance of the state on the subsystem $\rho_A$ to the equilibrium state, which is the maximally mixed state $\iden_A/d_A$, and we would like to bound the probability that the distance becomes large. For the state $\rho_A(t) = \tr_B U_t \ketbra{\psi} U_t^\dagger $ where $U_t$ is a depth $t$ random quantum circuit, we establish bounds of the form 
\begin{equation}
\Pr\left(\big\|\rho_A(t) - \iden_A/d_A\big\|_1 \geq \delta \right) \leq f(t)\,.
\end{equation}

We prove two versions of the bounds on fluctuations of entropy and distance from the maximally mixed state, one which captures the behavior at early times and another which captures the late-time behavior. The early-time bound for both entropy and distance fluctuations is a straightforward application of Markov's inequality and a computation of the averaged purity of the time-evolved state on the subsystem $\rho_A(t)$. 

The late-time bound on fluctuations utilizes the fact that random quantum circuits form unitary designs. Specifically, we can prove that entropy and distance fluctuations are suppressed for unitary $k$-designs, in a manner which depends on the design order $k$. 
Knowing the circuit depth at which random circuits form designs gives bounds on the fluctuations at late times.

\subsubsection*{Bounds on subsystem fluctuations for RQCs at early times}
Both the probability that the entropy of $\rho_A$ is submaximal and the probability that $\rho_A$ is far from the maximally mixed state can be bounded in terms of the expected value of the purity. For entropy fluctuations, using the monotonicity of entropies $S(\rho_A) \geq S_2(\rho_A)$, where $S_2(\rho_A) = -\log \tr(\rho_A^2)$ is the R\'enyi 2-entropy, we have
\begin{equation}
\Pr\big( S(\rho_A) \leq \log(d_A) - \delta \big) \leq \Pr \big( S_2(\rho_A) \leq \log(d_A) -\delta \big) = \Pr \Big( \tr(\rho_A^2) \geq \frac{1}{d_A}\,e^\delta \Big)\,.
\end{equation}
Further noting that the purity $\tr(\rho_A^2) \geq 1/d_A$, we can use Markov's inequality to upper bound the probability as
\begin{equation}
\Pr\big( S(\rho_A) \leq \log(d_A) - \delta \big) \leq \Pr \Big( d_A \tr(\rho_A^2) -1 \geq e^\delta-1 \Big) \leq \frac{d_A\BE\big[\tr(\rho_A^2)\big]-1}{e^\delta-1}\,.
\label{eq:ent2ndbound}
\end{equation}

Similarly, the probability that the state $\rho_A$ deviates away from the maximally mixed state can be re-expressed in terms of the purity in two ways. Recalling the quantum Pinsker inequality $\frac{1}{2}\|\rho-\sigma\|_1^2 \leq S(\rho\|\sigma)$, where $S(\rho\|\sigma) = -\tr(\rho\log\sigma) - S(\rho)$ is the relative entropy, we then have
\begin{equation}
\big\|\rho_A - \iden_A/d_A\big\|_1^2 \leq 2\big(\log(d_A) - S(\rho_A)\big)\leq 2\big(\log(d_A) - S_2(\rho_A)\big) = 2 \log( d_A \tr(\rho_A^2))\,. 
\label{eq:qmpinsker}
\end{equation}
In addition, we can also bound the trace distance by the 2-norm distance between states using the following relation between Schatten norms: $\|\op\|_1 \leq \sqrt{d} \|\op\|_2$\,, which gives
\begin{equation}
\big\|\rho_A - \iden_A/d_A\big\|_1^2 \leq d_A \big\|\rho_A - \iden_A/d_A\big\|_2^2 = d_A \tr(\rho_A^2)-1\,.
\label{eq:2norm}
\end{equation}
Proceeding, we can upper bound the probability that the trace distance is large again using Markov's inequality
\begin{equation}
\Pr\left(\big\|\rho_A - \iden_A/d_A\big\|_1\geq \delta \right)\leq \Pr\Big(2 \log( d_A \tr(\rho_A^2)) \geq \delta^2 \Big) = \Pr\Big( d_A \tr(\rho_A^2) -1\geq e^{\delta^2/2}-1 \Big) \leq \frac{d_A \BE[\tr(\rho_A^2)] - 1}{e^{\delta^2/2}-1}\,,
\label{eq:td2ndbound}
\end{equation}
as well as
\begin{equation}
\Pr\left(\big\|\rho_A - \iden_A/d_A\big\|_1\geq \delta \right)\leq \Pr\Big( d_A \tr(\rho_A^2) -1\geq \delta^2 \Big) \leq \frac{d_A\BE[\tr(\rho_A^2)] - 1}{\delta^2}\,.
\label{eq:td2ndbound2}
\end{equation}
The 2-norm bound gives a stronger upper bound on the probability for small fluctuations in the trace distance, but for very large fluctuations
the formulation in terms of Pinsker is stronger.

\begin{theorem*}[Restatement of \autoref{thm:earlytimebound}]
Assume $A$ is a contiguous subsystem with an even number of sites.
For depth $t$ brickwork random quantum circuits on a periodic 1D chain of qudits with local dimension $q$, and for some $\delta>0$, the entropy of the evolved state $\rho_A(t)$ obeys
\begin{equation}
\Pr\big(S(\rho_A(t) )\leq \log(d_A) -\delta\big) \leq \frac{1}{e^\delta-1} \left(\frac{d_A}{d_B} + d_A \left(\frac{2q}{q^2+1}\right)^{2(t-1)}\right)\,,
\end{equation}
and the trace distance of $\rho_A(t)$ to the maximally mixed state is
\begin{equation}
\Pr\left(\big\|\rho_A(t) - \iden_A/d_A\big\|_1\geq \delta \right)\leq \frac{1}{\max \{\delta^2, e^{\delta^2/2}-1\}} \left(\frac{d_A}{d_B} + d_A \left(\frac{2q}{q^2+1}\right)^{2(t-1)}\right)\,.
\end{equation}
\end{theorem*}

\subsubsection*{Proof of \autoref{thm:earlytimebound}}
Using the reformulation in terms of the purity by Markov's inequality in Eqs.~\eqref{eq:ent2ndbound}, \eqref{eq:td2ndbound}, and \eqref{eq:td2ndbound2}, we can leverage the expected value of the purity of $\rho_A(t)$ for depth $t$ random quantum circuits to bound the entropy and trace distance:

\begin{proposition}
\label{prop:RQCpurity}
For brickwork random quantum circuits on $n$ qudits with local dimension $q$, periodic boundary conditions, and evolved to depth $t$, the averaged purity of $\rho_A(t)$ on a contiguous subsystem $A$ is bounded as
\begin{equation}
\BE_{\rm RQC} \big[ \tr\,\rho_A(t)^2 \big] \leq \frac{1}{d_A} + \frac{1}{d_B} + \left(\frac{2q}{q^2+1}\right)^{2(t-1)}\,.
\end{equation}
For simplicity, we assume that the subsystem $A$ consists of an even number of qudits.
\end{proposition}
We delay the proof of this proposition until later in the \hyperref[app:purityproof]{Appendix}. 
The bounds on the entropy and distance fluctuations in \autoref{thm:earlytimebound} then follow directly from \autoref{prop:RQCpurity}. \qed

\subsubsection*{Unitary designs}
Before we proceed to prove \autoref{thm:appdesignbound}, concentration of the entropy and distance for approximate unitary designs, we first recall some definitions. The $k$-fold channel of an operator $\op$ with respect to the ensemble of unitaries $\CE$ is
\begin{equation}
\Phi_\CE^{(k)}(\op) := \int_\CE dU\, U^{\otimes k}(\op) U^\dagger{}^{\otimes k}\,.
\end{equation}
For the channel with respect to the Haar measure on the entire unitary group $U(d)$, we write $\Phi_U^{(k)}$.
The diamond norm of a channel is defined as
\begin{equation}
\| \Phi \|_\diamond := \sup_{\psi, d} \big\| \Phi \otimes \mathcal{I}_d (\ketbra{\psi}) \big\|_1\,.
\end{equation}
\begin{definition}[Approximate unitary design]
\label{def:approxkdesign}
For $\ep>0$, an ensemble of unitaries $\CE$ is an $\ep$-approximate unitary $k$-design if the distance between the $k$-fold channels is bounded as
\begin{equation}
\left\| \Phi_\CE^{(k)}-\Phi_U^{(k)} \right\|_\diamond \leq \ep\,.
\end{equation}
\end{definition}

We now want to establish bounds on the entropy and trace distance for approximate unitary designs. Concentration bounds for designs have already been studied in \cite{LowDeviation09} for general monomials, but also for both the von Neumann entropy of a random state and the distance to the canonical state, see their Theorems 1.3 and 1.4. In directly bounding centered moments of the purity, we improve on the large deviation results in \cite{LowDeviation09}, but some aspects of our proof still closely mirror the steps taken there.

\begin{theorem*}[Restatement of \autoref{thm:appdesignbound}]
For an $\ep$-approximate unitary $4k$-design $\CE$, where $\ep$ is taken to be $1/d^{2k}$, the entropy $S(\rho_A)$ of $\rho_A = \tr_B( U\ketbra{\psi}U^\dagger)$, where $U$ is drawn from $\CE$, obeys
\begin{equation}
\Pr\big( S(\rho_A) \leq \log(d_A) - \delta \big) \leq 2(k!+1/d^k) \left(\frac{9\pi^3}{\gamma^2} \frac{d_A}{d_B}\right)^k\,,
\end{equation}
where $\gamma:= e^\delta -1 -\frac{d_A}{d_B}$ and $\delta $ must be taken to be $\delta \geq \frac{d_A}{d_B}$. 
Similarly, the distance of $\rho_A$ to the maximally mixed state $\iden_A/d_A$ obeys
\begin{equation}
\Pr\Big( \big\|\rho_A - \iden_A/d_A\big\|_1 \geq \delta\Big) \leq 2(k!+1/d^k) \left(\frac{9\pi^3}{\eta^2} \frac{d_A}{d_B}\right)^k\,,
\end{equation}
where $\eta:= \max \{\delta^2, e^{\delta^2/2} -1\} -\frac{d_A}{d_B}$, for which we must take $\delta^2 > \frac{d_A}{d_B}$.
\end{theorem*}

\subsubsection*{Proof of \autoref{thm:appdesignbound}}
\label{app:thm2proof}
We now consider the entropy of the state $\rho_A = \tr_B(U\ketbra{\psi}U^\dagger)$, where $U$ is drawn at random from an approximate unitary design $\CE$ and $\ket\psi$ is an arbitrary initial state. The probability of the entropy being small over the ensemble can be bounded as
\begin{align}
\Pr\big( S(\rho_A) \leq \log(d_A) - \delta \big) & \leq \Pr \Big( \tr(\rho_A^2) \geq \frac{1}{d_A}\,e^\delta \Big) \\ 
& \leq \Pr \bigg( \Big(\tr(\rho_A^2)- \BE_U\big[ \tr(\rho_A^2)\big] \Big)^{2k} \geq \Big(\frac{1}{d_A}e^\delta- \BE_U\big[ \tr(\rho_A^2)\big] \Big)^{2k} \bigg)\,,
\end{align}
where second inequality requires $\frac{1}{d_A}e^\delta- \BE_U\big[ \tr(\rho_A^2)\big] \geq 0$. 
We have subtracted the Haar average of the purity in anticipation of applying Levy's lemma. Direct calculation gives that $ \BE_U\big[ \tr(\rho_A^2)\big] = \frac{d_A + d_B}{d+1} $, so requiring $\delta \geq \frac{d_A}{d_B}$ is sufficient for the former requirement.  Markov's inequality then gives
\begin{equation}
\Pr\big( S(\rho_A) \leq \log(d_A) - \delta \big) \leq \left(\frac{d_A}{\gamma}\right)^{2k} \BE_\CE \Big[\big(\tr(\rho_A^2)- \BE_U\big[ \tr(\rho_A^2)\big] \big)^{2k} \Big] \where \gamma := e^\delta - 1 -\frac{d_A}{d_B}
\end{equation}
and we have used that $\frac{d_A + d_B}{d+1} \leq \frac{1}{d_A} + \frac{1}{d_B}$. Note that we took each side of the inequality to the power $2k$ to ensure that both sides are positive. The assumption $\delta \geq \frac{d_A}{d_B}$ ensures $\gamma>0$. 
Using \autoref{prop:puritymoms} to bound the centered moments of the purities, we find that for an $\ep$-approximate $4k$-design with $\ep=1/d^{2k}$, the probability of the entropy is bounded as
\begin{equation}
\Pr\big( S(\rho_A) \leq \log(d_A) - \delta \big) \leq 2k! \left(\frac{9\pi^3}{\gamma^2} \frac{d_A}{d_B}\right)^k + \frac{1}{d^k} \left( \frac{(1+\frac{1}{d_A}+\frac{1}{d_B})^2}{\gamma^2 } \frac{d_A}{d_B}\right)^{k}\,.
\end{equation}
Simplifying the expression, we can write
\begin{equation}
\Pr\big( S(\rho_A) \leq \log(d_A) - \delta \big) \leq (2k!+1/d^k) \left(\frac{9\pi^3}{\gamma^2} \frac{d_A}{d_B}\right)^k\,.
\end{equation}

Proceeding similarly for the state distance, we can upper bound the probability that the trace distance between $\rho_A$ and the maximally mixed state is large using Eq.~\eqref{eq:qmpinsker} and Eq.~\eqref{eq:2norm} to write
\begin{align}
\Pr\left( \big\|\rho_A - \iden_A/d_A \big\|_1\geq \delta\right) &\leq \Pr\left( d_A \tr(\rho_A^2) \geq \max\big\{\delta^2 + 1, e^{\delta^2/2}\big\} \right)\\
&\leq 
\Pr\left( \Big(\tr(\rho_A^2) - \BE_U\big[\tr(\rho_A^2)\big]\Big)^{2k} \geq \Big(\frac{1}{d_A} \max\big\{\delta^2 + 1, e^{\delta^2/2}\big\} - \BE_U\big[\tr(\rho_A^2)\big]\Big)^{2k}\right)\,,
\end{align}
where just as with the entropies, we center the moments of the purity. We can again apply Markov's inequality to bound the trace distance fluctuations as
\begin{equation}
\Pr\left( \left\|\rho_A - \frac{\iden_A}{d_A}\right\|_1\geq \delta\right) \leq 
\left(\frac{d_A}{\eta}\right)^{2k} \BE_\CE\Big[\big(\tr(\rho_A^2) - \BE_U\big[\tr(\rho_A^2)\big]\big)^{2k}\Big] \where \eta := \max\big\{\delta^2, e^{\delta^2/2}-1\big\} -\frac{d_A}{d_B}\,.
\end{equation}
Using \autoref{prop:puritymoms} to bound expectation of the purities, we find that for an $\ep$-approximate $4k$-design with $\ep=1/d^{2k}$
\begin{equation}
\Pr\left( \left\|\rho_A - \frac{\iden_A}{d_A}\right\|_1\geq \delta\right) \leq (2k!+1/d^k) \left(\frac{9\pi^3}{\eta^2} \frac{d_A}{d_B}\right)^k\,,
\end{equation}
which completes the proof. \qed

\begin{proposition}\label{prop:puritymoms}
For an $\ep$-approximate $4k$-design $\CE$, centered moments of the purity of $\rho_A = \tr_B (U\ketbra{\psi}U^\dagger)$, with $U$ drawn randomly from $\CE$, obey
\begin{equation}
\BE_\CE \Big[\big(\tr(\rho_A^2)- \BE_U\big[ \tr(\rho_A^2)\big] \big)^{2k} \Big] \leq 2 k! \left(\frac{9\pi^3}{d}\right)^k + \ep \left(1+\frac{1}{d_A}+\frac{1}{d_B}\right)^{2k}\,.
\end{equation}
\end{proposition}

\begin{proof}
We need to establish bounds for approximate unitary designs of the form $\BE_\CE [(\tr(\rho_A^2)- c )^{2k} ]$. Proceeding, we write
\begin{equation}
\BE_\CE \Big[\big(\tr(\rho_A^2)- c \big)^{2k} \Big] = \BE_\CE \Big[\big(\tr(\rho_A^2)- c \big)^{2k} \Big]  - \BE_U \Big[\big(\tr(\rho_A^2)- c \big)^{2k} \Big] + \BE_U \Big[\big(\tr(\rho_A^2)- c \big)^{2k} \Big]\,,
\label{eq:designbound}
\end{equation}
where $\BE_U[\,\cdot\,]$ denotes the expectation with respect to the Haar measure on $U(d)$.
First, we prove that for an $\ep$-approximate unitary $4k$-design $\CE$ and any $n\leq 2k$, we have
\begin{equation}
\BE_\CE \big[(\tr(\rho_A^2))^n \big]  - \BE_U \big[(\tr(\rho_A^2))^n \big] \leq \ep\,.
\label{eq:binomialdifs}
\end{equation}
Noting that the purity of $\rho_A = \tr_B(\ketbra{\psi})$ can be written as 
\begin{equation}
\tr(\rho_A^2) = \tr\,\big( \ketbra{\psi}^{\otimes 2} (\Pi_A\otimes \iden_B)\big)\,,
\end{equation}
where $\iden_B$ is the identity operator on the $B$ subsystem of the 2-fold space and $\Pi_A$ is the swap operator on the $A$ subsystem of the 2-fold space. Then we have
\begin{align}
\BE_\CE \big[(\tr(\rho_A^2))^n \big]  - \BE_U \big[(\tr(\rho_A^2))^n \big] &= \tr \left(\!\left(\int_\CE dU (U\ketbra{\psi}U^\dagger)^{\otimes 2n} - \int_U dU (U\ketbra{\psi}U^\dagger)^{\otimes 2n}\right)\! \big(\Pi_A\otimes \iden_B\big)^{\otimes n}\!\right) \\
&=\tr \left( \Big(\Phi_\CE^{(2n)}\big(\ketbra{\psi}^{\otimes 2n}\big) - \Phi_U^{(2n)}\big(\ketbra{\psi}^{\otimes 2n}\big)\Big) \big(\Pi_A\otimes \iden_B\big)^{\otimes n} \,\right)\\
&\leq\left\| \Big(\Phi_\CE^{(2n)}\big(\ketbra{\psi}^{\otimes 2n}\big) - \Phi_U^{(2n)}\big(\ketbra{\psi}^{\otimes 2n}\big)\Big) \big(\Pi_A\otimes \iden_B\big)^{\otimes n} \right\|_1 \\
&\leq \left\| \Big(\Phi_\CE^{(2n)} - \Phi_U^{(2n)}\Big) \big(\ketbra{\psi}^{\otimes 2n}\big)\right\|_1 \left\|\Pi_A^{\otimes n}\otimes \iden_B^{\otimes n}\right\|_\infty\\
&\leq \left\| \Phi_\CE^{(2n)} - \Phi_U^{(2n)} \right\|_\diamond\,,
\end{align}
where in the second to last line we use H\"older's inequality, and in the last line we used the definition of the diamond norm and that the operator norm of the permutation operator is one. 
Now we observe that for $n\leq 2k$
\begin{equation}
\BE_\CE \big[(\tr(\rho_A^2))^n \big]  - \BE_U \big[(\tr(\rho_A^2))^n \big] \leq \left\| \Phi_\CE^{(2n)} - \Phi_U^{(2n)} \right\|_\diamond \leq \left\| \Phi_\CE^{(4k)} - \Phi_U^{(4k)} \right\|_\diamond \leq \ep\,,
\end{equation}
as an approximate $k$-design is also an approximate $(k-1)$-design, and where we have used the definition of an $\ep$-approximate $4k$-design. This proves the claim in Eq.~\eqref{eq:binomialdifs}. Proceeding, we can use this bound to establish that
\begin{align}
\BE_\CE \Big[\big(\tr(\rho_A^2)- c \big)^{2k} \Big]  - \BE_U \Big[\big(\tr(\rho_A^2)- c \big)^{2k} \Big] &= \sum_{n=0}^{2k} \binom{2k}{n} \Big( \BE_\CE \big[(\tr(\rho_A^2))^n \big]  - \BE_U \big[(\tr(\rho_A^2))^n \big]\Big) (-c)^{2k-n}\\
&\leq \ep (1+c)^{2k}\,,
\label{eq:diffbound}
\end{align}
where we take norm of the RHS in the first line to upper bound the difference of expectations. 

Lastly, we bound the Haar random expectation $\BE_U [(\tr(\rho_A^2)- c)^{2k}]$, using the integral form of the expectation and Levy's lemma. This nice method of bounding centered moments by integrating the concentration bound appeared in \cite{LowDeviation09,Bellare94}.
Recall that Levy's lemma says if $f(\ket{\psi})$ is an $L$-Lipschitz function on the Bloch sphere and $\ket\psi$ is chosen uniformly at random, then
\begin{equation}
\Pr \Big( \big| f(\ket{\psi}) - \BE_{\ket\psi} [f(\ket{\psi})]\big| \geq \delta \Big) \leq 2 \exp\left( - \frac{4d\delta^2}{9\pi^3 L^2} \right)\,.
\end{equation}
Then it follows that as $\tr (\tr_B \ketbra{\psi})^2$ is a 2-Lipschitz function on the complex unit sphere 
\begin{align}
\BE_U \Big[\big(\tr(\rho_A^2)- \BE_U[\tr(\rho_A^2)]\big)^{2k}\Big] &= \int_0^\infty dy\, \Pr\big(\big|\tr(\rho_A^2)- \BE[\tr \,\rho_A^2]\big| \geq y^{1/2k}\big)\\
&\leq 2 \int_0^\infty dy\, \exp\bigg( -\frac{dy^{1/k}}{9\pi^3} \bigg) = 2 (k!)\left(\frac{9\pi^3}{d}\right)^k\,.
\label{eq:Levybound}
\end{align}
Altogether, using Eq.~\eqref{eq:diffbound} and Eq.~\eqref{eq:Levybound} in Eq.~\eqref{eq:designbound} we have that if $\CE$ forms an $\ep$-approximate unitary $4k$-design
\begin{equation}
\BE_\CE \Big[\big(\tr(\rho_A^2)- \BE_U\big[ \tr(\rho_A^2)\big] \big)^{2k} \Big] \leq 2 (k!) \left(\frac{9\pi^3}{d}\right)^k + \ep \left(1+\frac{d_A+d_B}{d+1}\right)^{2k}\,,
\end{equation}
from which the proposition follows.
\end{proof}

As the above theorem bounds entropy fluctuations when the design order is greater than 4, we can separately consider the bound for approximate 2-designs.
\begin{proposition} For an $\ep/d$-approximate 2-design $\CE$, the entropy of $\rho_A = \tr_B U\ketbra{\psi}U^\dagger$ with $U$ drawn from $\CE$, is close to maximal with probability
\begin{equation}
\Pr\big( S(\rho_A)\leq \log(d_A) - \delta\big) \leq \frac{d_A}{d_B} \frac{1+\ep}{e^\delta-1}\,,
\end{equation}
and the state $\rho_A$ is close to maximally mixed with probability
\begin{equation}
\Pr\left( \big\|\rho_A - \iden_A/d_A \big\|_1\geq \delta \right) \leq \frac{d_A}{d_B}\frac{1+\ep}{\max\{\delta^2,e^{\delta^2/2}-1\}}\,.
\end{equation}
\end{proposition}
\begin{proof}
The bounds on entropy fluctuations follow from Eqs.~\eqref{eq:ent2ndbound} and \eqref{eq:td2ndbound} and a bound on the averaged purity for $\ep$-approximate 2-designs in Eq.~\eqref{eq:binomialdifs} with $n=1$.
\end{proof}

\subsubsection*{Random quantum circuits form approximate unitary designs}
Having established concentration bounds for approximate unitary designs, we can use previous results, which give the depth at which the set of random circuits form designs, to bound entropy and subsystem fluctuations as a function of time. This approach of rigorously exploring the late-time behavior of RQCs by utilizing high-degree unitary designs was also taken in \cite{compgrowth19}, which proved a long-time growth of quantum complexity for random quantum circuits.

Local random quantum circuits are known to be efficient constructions of approximate unitary designs \cite{HL08,BH13,BHH12}.
Specifically, Brand{\~a}o, Harrow, and Horodecki proved that local and parallelized random quantum circuits form approximate designs after some depth which is polynomial in $k$. We review their result for systems of local qubits. 
\begin{theorem}[\cite{BHH12}] \label{thm:BHH}
For $d=2^n$, $k\leq \sqrt{d}$, and $\ep>0$, the set of all 1D random quantum circuits on $n$ qubits with Haar random 2-local gates drawn from $U(4)$, forms an $\ep$-approximate unitary $k$-design if the circuit depth is
\begin{equation}
    t\geq C \lceil \log(k) \rceil^2 k^{9.5} (nk +\log(1/\ep))\,,
\end{equation}
where $C$ is a constant computed in \cite{BHH12}.
\end{theorem}
\ni The above bound can be simplified as follows:
\begin{corollary}\label{cor:rqcdesigns}
1D brickwork random quantum circuits on $n$ qubits form $\ep=1/d^k$-approximate unitary $k$-designs when $t\geq c_{\rm bw} n k^{11}$, where $c_{\rm bw}$ is taken to be $4.6 \times 10^7$.
\end{corollary}

They also extend their result to random circuits comprised of gates drawn from some universal gate set $G$, where the $C$ is then a (potentially large) constant depending on $G$. We note that the RQCs we have primarily been considering are of the brickwork type, alternating layers of 2-local unitaries on even and odd links. The parallelized model in \cite{BHH12} is slightly different, applying even and odd layers with equal probability at each time step. But this parallelized model mixes slower than the brickwork RQCs and upper bounds the design depth. Thus, we can extend the above theorem to brickwork RQCs. 

The $t=O(nk^{11})$ behavior in \autoref{thm:BHH} for random circuits on local qubits can be improved to $t=O(nk)$ by taking the local dimension to be large. 
\begin{theorem*}[Restatement of \autoref{thm:lineardesigngrowth}]
Brickwork random quantum circuits on $n$ qudits of local dimension $q$ form $\ep$-approximate unitary $k$-designs if the circuit depth is $t\geq 2nk + \log_q (1/\ep)$, for some large value of $q$ which depends on $k$ and the size of the circuit.
\end{theorem*}

A consequence of these two theorems are time-dependent bounds on the fluctuations of subsystem entropy and the trace distance to the equilibrium state discussed in \autoref{sec:RQC_analytics}. For RQCs on local qubits, \autoref{thm:BHH} gives that $t = c_{\rm bw} n (4k)^{11}$ depth RQCs form $\ep$-approximate $4k$-designs with $\ep=1/d^{2k}$. \autoref{thm:appdesignbound} then tells us that the entropy fluctuations of a $d_A= O(1)$ sized subsystem are bounded for depth $t$ circuits as
\begin{equation}
    \Pr\big( S(\rho_A(t)) \leq \log (d_A) - \delta \big) \lesssim \left(\frac{t^{1/11}}{e^{2\delta}}\frac{1}{d} \right)^{(t/n)^{1/11}}\,,
\end{equation}
which holds up to exponential times for $t\lesssim n2^{11n/2}$. For large local dimension, \autoref{thm:lineardesigngrowth} gives that that RQCs form $\ep=1/d^{2k}$ approximate $4k$-designs when the circuit depth is $t=10nk$. Subsequently, \autoref{thm:appdesignbound} bounds the probability the subsystem entropy of a state evolved by a depth $t$ RQC deviates from maximal as
\begin{equation}
    \Pr\big( S(\rho_A(t)) \leq \log (d_A) - \delta \big) \lesssim \left(\frac{t}{e^{2\delta}n}\frac{1}{d}\right)^{t/n}\,,
\end{equation}
but as we must take $q$ large to achieve this linear behavior, we cannot extend the bound to exponential times. Similar time-dependent bounds on the distance of the evolved state $\rho_A(t)$ to the maximally mixed state follow from \autoref{thm:appdesignbound}.

Lastly, we note that the pseudorandomness properties have been studied in other random circuit models, and it is known that Brownian random circuits \cite{Onorati17}, higher-dimensional random circuits \cite{HM18}, Clifford random circuits with a small number of non-Clifford gates \cite{haferkamp2020homeopathy}, and some time-dependent Hamiltonian constructions \cite{Nakata16} form approximate unitary designs. Combining \autoref{thm:appdesignbound} with their design results gives bounds on the entropy and subsystem fluctuations in these models.

\subsubsection*{Counting states of a given entropy}
Above we gave a probabilistic statement about the likelihood the entropy of a subsystem of a state evolved by a random circuit had fluctuated away from its equilibrium value. We found that with extremely high probability, the entropy was close to maximal after the thermalization time, and the suppression of the fluctuations continued up until exponential times. It is also natural to ask how many states, in the set of states generated by depth $t$ random quantum circuits, have a given entropy. Following \cite{compgrowth19}, we can turn our probabilistic statement about entropy into a quantitative one using a bound on the weights of an approximate design. To discuss this in a concrete setting, where the set of states at time $t$ is finite, we consider a slightly different random circuit model.

We now consider $G$-local random quantum circuits on $n$ qudits, identical to the brickwork RQCs discussed above but instead each 2-local gate is chosen randomly from a universal gate set $G$. The number of states at time $t$, generated by depth $t$ $G$-local RQCs, is upper bounded up $|G|^{n t}$, where $|G|$ is the cardinality of the gate set. As was proved in \cite{BHH12}, $G$-local random quantum circuits also form $\ep$-approximate unitary designs in a depth $t=c(G) k^{10}(nk+\log1/\ep)$, where $c(G)$ is a potentially large constant depending on the gate set $G$. 

For the ensemble of states $\CE_\psi = \{p_i, \ket{\psi_i}\}$ generated by a discrete $1/d^{2k}$-approximate unitary $4k$-design, i.e. by acting on a fixed state with the unitaries of the design, \autoref{thm:appdesignbound} gives that
\begin{equation}
\Pr\big( S(\rho_A) > \log(d_A) - \delta \big) \geq 1- 2\big(k!+1/d^k\big) \left(\frac{9\pi^3}{\gamma^2} \frac{d_A}{d_B}\right)^k\,.
\end{equation}
We can also write the probability of the event as a sum over the elements of the ensemble of states as
\begin{equation}
\Pr\big( S(\rho_A) > \log(d_A) - \delta \big) = \sum_i p_i \mathbf{1}\{S(\rho_A) > \log(d_A) - \delta\} \leq N \frac{2}{d^{2k}}\,,
\end{equation}
where $\mathbf{1}$ is the indicator function and $N$ denotes the number of states in the ensemble $\CE_\psi$ with entropy $S(\rho_A) > \log(d_A) - \delta$. To get the upper bound we used an upper bound on the weights of a design in \autoref{lem:stateweights}. Together, we find that the number of distinct states with entropy $S(\rho_A) > \log(d_A) - \delta$ in the set of states generated by an approximate unitary $4k$-design is
\begin{equation}
N \geq \frac{d^{2k}}{2}\left(1- 2\big(k!+1/d^{2k}\big) \left(\frac{9\pi^3}{\gamma^2} \frac{d_A}{d_B}\right)^k\right)\,.
\end{equation}
The number of states is exponential as long as $\delta\gtrsim \sqrt{d_A/d_B}$.
This proves that there are at least $\Omega(q^{2nk})$ distinct states with $S(\rho_A) \gtrsim \log(d_A) - \sqrt{d_A/d_B}$. For local random quantum circuits which form approximate unitary designs in a depth $t = O(nk)$, there are then
$N \gtrsim e^{t}$ distinct states with $S(\rho_A(t)) \gtrsim \log(d_A) - \sqrt{d_A/d_B}$\,. 

\begin{lemma}[Lemma 1 in \cite{compgrowth19}]\label{lem:stateweights}
The weights of a discrete $\ep$-approximate complex projective $k$-design obey
\begin{equation}
p_i\leq \frac{k!}{d^k}+\ep\,.
\end{equation}
\end{lemma}

\subsubsection*{Purity bound for random circuits}
\label{app:purityproof}
We now prove the early time purity bound for brickwork random quantum circuits, which gave us bounds on the early time fluctuations. The early-time decay of the purity for RQCs has been studied in \cite{nahum2017entgrowth,ZNstatmech19,bertini2020scrambling}.

\begin{proposition*}[Restatement of \autoref{prop:RQCpurity}]
For brickwork random quantum circuits on $n$ qudits with local dimension $q$, periodic boundary conditions, and evolved to depth $t$, the averaged purity of $\rho_A(t)$ on a contiguous subsystem $A$, consisting of an even number of qudits, is bounded as\footnote{\mbox{NHJ would like to thank {\'A}lvaro Alhambra for helpful discussions on this proof.}}
\begin{equation}
\BE_{\rm RQC} \big[ \trs\rho_A(t)^2 \big] \leq \frac{1}{d_A} + \frac{1}{d_B} + \left(\frac{2q}{q^2+1}\right)^{2(t-1)}\,.
\end{equation}
\end{proposition*}

\begin{proof}
We are interested in bounding the purity of the evolved state $\rho_A(t) = \tr_B(U_t\rho U_t^\dagger)$ for a contiguous subsystem $A$, where $U_t$ is a depth $t$ brickwork random circuit. In the following we will use $A$ not only to label the subsystem but also to denote the number of qudits in $A$.

It has been well-established that the purity of a subsystem of a Haar random state is very close to minimal. But in \cite{nahum2017entgrowth} it was noted that for random circuits, the purity for a single bipartition, \ie of the half-line, obeys a simple finite-difference equation.
Consider a bipartition of the system at site $x$ across a 2-local gate in the $t$-th layer of the circuit. By averaging that layer of the circuit, the purity at time $t$ can be related to the purity at $t-1$ as
\begin{equation} \label{eq:NVRH}
    \trs (\rho_x(t)^2) = \frac{q}{q^2+1} \big( \trs (\rho_{x+1}(t-1)^2) + \trs (\rho_{x-1}(t-1)^2)\big)\,.
\end{equation}
We will consider the derivation of the above equation, found in Appendix B of \cite{nahum2017entgrowth}, as background for the following derivation of \autoref{prop:RQCpurity}.

Eq.\ \eqref{eq:NVRH} re-frames the computation of purity as a stochastic process starting at the $t$-th layer and extending back to time zero. In a tensor network picture, we can interpret the finite-difference equation as an evolution rule telling us that the cut at $x$ (at circuit time $t$) moves left or right after the evolution by the layer (at circuit time $t-1$). 
Iterating back to the initial state, the calculation of the purity for the half-line is simply the sum of all paths moving through the circuit that start at the cut $x$. Thus it is a random walk problem, where the random walker starts at $x$ at the $t$-th layer and can move either one step left or right at every time step, with a weight per time step given by $q/(q^2+1)$.
If the initial state is a pure product state $\ket\psi = \ket{\psi_1}\otimes \cdots\otimes \ket{\psi_n}$, then all subsystem purities at $t=0$ are one and the RQC averaged purity for a single cut is
\begin{equation}
\BE_{\rm RQC} \big[\tr(\rho_x (t)^2)\big] = \sum_{\rm paths} \left(\frac{q}{q^2+1}\right)^t \,.
\end{equation}
For arbitrary initial states the above is an upper bound as we may simply upper bound the purities by one. 

Now we want to consider the calculation of the evolution of the purity of an interval $A$, consisting of $A$ sites and thus with dimension $q^A$. For simplicity, assume the interval contains an even number of sites, and that both boundaries of the region are cuts across 2-local gates of the $t$-th layer. We now have two entanglement cuts in the circuit, at the ends of the interval. The calculation of the purity will then become a sum over the configurations of two random walkers starting at the boundaries of $A$, moving through the circuit. If two random walkers meet after some number of time steps, they can annihilate and give a finite non-decaying contribution to the purity sum.

For times $t<A/2$, the paths will not cross and the purity is simply the weighted sum over the $2^{2t}$ paths
\begin{equation}
\BE_{\rm RQC} \big[\tr(\rho_A (t)^2)\big] =  \left(\frac{2q}{q^2+1} \right)^{2t}\,.
\end{equation}
After time $t=A/2$, the random walkers starting from the two boundaries may annihilate, giving a finite contribution to the sum at that time step. For instance, at time $t=A/2$, exactly one of the $2^{A}$ possible paths from the ends of $A$, intersects in the middle of the subsystem. At the next time step there will be $O(A)$ paths which intersect. 
We have defined the brickwork random circuits with periodic boundary conditions, evolving a ring of qudits. Thus, we get contributions from the two ways our random walkers can intersect, through the $A$ subsystem and through the $B$ subsystem, as well as over the possible non-intersecting paths up to time $t$.

We have reduced the calculation of the expected purity for random quantum circuits to a combinatorial problem of enumerating random walk configurations. The exact expression for the averaged purity is
\begin{equation}
\BE_{\rm RQC} \big[\tr(\rho_A (t)^2)\big] =  \sum_{t'=1}^t c_A(t') \left(\frac{q}{q^2+1} \right)^{2t'} + \sum_{t'=1}^t c_B(t') \left(\frac{q}{q^2+1} \right)^{2t'} + g(t) \left(\frac{q}{q^2+1} \right)^{2t}\,,
\label{eq:purityeqrw}
\end{equation}
where $c_A(t')$ is the number of intersections at a time step $t'$ by two random walkers separated by a distance of $A$, similarly for $c_B(t')$, and $g(t)$ is the number of configurations of two paths at time $t$ with no previous crossings, i.e. the number of possible ways the random walkers can reach the $t=0$ boundary.

To compute the coefficient $c_A(t)$ we ask how many ways can two non-intersecting random walkers in 1D, separated by a distance of $A$ sites at time $t=0$, can meet at time $t$, i.e. assuming the random walkers have not intersected at any previous time step.\footnote{\begin{minipage}[t]{2\linewidth}{In the statistical mechanics literature, this is sometimes called a reunion of viscous random walkers.}\end{minipage}}  This problem can be solved using a method of images for random walks \cite{Fisher84,Huse84}, to account for the constraint that the paths the walkers take cannot cross, and the result we find is
\begin{equation}
c_A(t) = \frac{A}{2 t} \binom{2t}{t-A/2}\,.
\end{equation}
In the expression of the purity in Eq.~\eqref{eq:purityeqrw}, the first sum over intersections of the two random walks at previous time steps can be upper bounded as
\begin{equation}
    \sum_{t'=1}^t \frac{A}{2 t'} \binom{2t'}{t'-A/2} \left(\frac{q}{q^2+1} \right)^{2t'} \leq \sum_{t'=1}^\infty \frac{A}{2 t'} \binom{2t'}{t'-A/2} \left(\frac{q}{q^2+1} \right)^{2t'} = \frac{1}{q^A}\,.
    \label{eq:sumbound}
\end{equation}
An equivalent bound gives that $c_B(t)\leq 1/q^B$. Noting that the coefficient $g(t)$, counting the number of possible paths for the two random walkers after $t$ time steps, is trivially upper bounded as $g(t)\leq 2^{2t}$, we conclude that
\begin{equation}
\BE_{\rm RQC} \big[\tr(\rho_A (t)^2)\big] \leq \frac{1}{q^A} + \frac{1}{q^B} + \left(\frac{2q}{q^2+1} \right)^{2t}\,.
\end{equation}
The bound holds for a contiguous subsystem $A$ consisting of an even number of sites and with boundaries on links, i.e. gates are applied across the edges of $A$ at the $t$-th time step. If the boundaries fall between gates, then the last layer of the random circuit does nothing and the random walk starts at $t-1$. Modifying the exponent accordingly to account for both cases, this completes the proof.
\end{proof}

\subsubsection*{Proof of \autoref{thm:counting}}
\label{app:cor2proof}

Here we prove \autoref{thm:counting} from the main text.  To do so, we require \autoref{thm:earlytimebound}, \autoref{thm:appdesignbound}, and further need to establish that random circuits form unitary $(k=d)$-designs.
The results from both \cite{BHH12} and \cite{NHJ19} do not extend to exponentially high-degree designs; in the former the proof is limited to $k\leq \sqrt{d}$ and in the latter a large $q$ limit is taken in a $k$-dependent way, making the $k=d=q^n$ regime inaccessible. Fortunately, a lower bound on the spectral gap given in \cite{BHH12} suffices to prove a bound on the depth for random circuits in the high-degree design regime and allows us to establish the following proposition.

\begin{proposition}
\label{prop:d5}
Brickwork random quantum circuits on $n$ qudits with local dimension $q$ form $\ep$-approximate unitary $k$-designs of exponentially high design degree $k=d=q^n$ and with $\ep=1/d^k$ if the circuit depth is
\begin{equation}
    t\geq 32 n^2 d^5\,.
\end{equation}
\end{proposition}
\begin{proof}
This proposition follows from a number of Lemmas in \cite{BHH12}. In that paper, the design depth is computed for random circuits by lower bounding the spectral gap of a frustration-free Hamiltonian. Using a path-coupling method for bounding the mixing time of random walks on the unitary group, \cite[Lemmas 19 and 20]{BHH12} established a $k$-independent lower bound on the spectral gap, $\Delta(H) \geq n^{-1}(e(q^2+1))^{-n}$, which is exponentially small in $n$ but this will suffice for our purposes.

Directly combining Lemmas 19 and 20 with the relation between the spectral gap and the bound on the distance to forming a design in Lemma 16, and extending the result from local RQCs to brickwork RQCs, gives the depth at which exponentially deep random circuits form high-degree designs. Taking $\ep=1/d^k$ and $k=d$ then establishes the claim.
\end{proof}

With this Proposition at hand, we can proceed with the proof of \autoref{thm:counting}.
\begin{theorem*}[Restatement of \autoref{thm:counting}]
For brickwork random quantum circuits on a 1D periodic chain of qubits, let $N_A^{\rm ent}(\delta)$ be the number of discrete times $t$ that a contiguous subsystem $A$ satisfies $S(\rho_A(t)) \leq \log(d_A) - \delta$ for all times in the range $c_{\rm th} \log(d_A) \leq t \leq e^{c_{\rm rec} d}$, where $c_{\rm th} > 1$ and $c_{\rm rec} < 1$.  Then for $n \geq \Omega(c_{\rm th} \log(d_A))$ and the constant $c_{\rm rec} = \gamma^2/(9 \pi^3 d_A^2 e)$, the probability of an entropy fluctuation is bounded as
\begin{equation}
\Pr\left( N^{\rm ent}_A(\delta)>0 \right) \leq \frac{8}{e^{\delta}- 1 }\,\left(\frac{1}{d_A}\right)^{\frac{2}{5}c_{\rm th} -1}.
\label{eq:counting1}
\end{equation}
Similarly, if $N_A^{\rm dist}(\delta)$ is the number of discrete times $t$ that $A$ satisfies $\|\rho_A(t) - \iden_A/d_A\|_1 \geq \delta$ for all $t$ in the same range, then for $n\geq \Omega(c_{\rm th} \log(d_A))$
\begin{equation}
\Pr\left( N^{\rm dist}_A(\delta)>0 \right) \leq \frac{8}{\max\{\delta^2,e^{\delta^2/2}- 1\}}\,\left(\frac{1}{d_A}\right)^{\frac{2}{5}c_{\rm th} -1}.
\label{eq:counting2}
\end{equation}
\end{theorem*}
We proven this for RQCs on local qubits ($q=2$) and note that the bound improves for larger local dimension.
\begin{proof}
We will prove that the probability, over random circuits, of a fluctuation from the thermalization time up to the recurrence time is bounded as in Eq.\ \eqref{eq:counting1}; the derivation of the 1-norm version in Eq.\ \eqref{eq:counting2} is essentially the same. First, note that the probability of a nonzero $N^{\rm ent}_A(\delta)$ can be written as a union of events:
\begin{equation}
\Pr\left( N^{\rm ent}_A(\delta)>0 \right) = \Pr\left( \bigcup\nolimits_{t=c_{\rm th}\log(d_A)}^{e^{c_{\rm rec} d}} \{S(\rho_A(t))\geq \log(d_A) - \delta\} \right)\,.
\end{equation}
We would like to bound the probability of an entropy fluctuation of a region $A$ (by more than an amount $\delta$) within the time range $T = [c_{\rm th} \log(d_A), e^{c_{\rm rec} d}]$.  We find it convenient to subdivide $T$ into four intervals:
\begin{equation}
T_1 = [c_{\rm th} \log(d_A), \,c_8 n]\,, \qquad T_2 = [c_8 n, c_{40} n]\,, \qquad T_3 = [c_{40} n, 32 n^2 d^5]\,, \qquad T_4 = [32 n^2 d^5, \,e^{c_{\rm rec} d}]\,,
\end{equation}
an early-time interval where we can apply \autoref{thm:earlytimebound} and three intervals extending from the scrambling time to the recurrence time where we can apply increasing stronger iterations of \autoref{thm:appdesignbound}, to bound the probability of a fluctuation in each interval. From \autoref{cor:rqcdesigns} it follows that 1D brickwork RQCs form an $\ep$-approximate unitary $4k$-design with $\ep=1/d^{2k}$ when the circuit depth is $t \geq c_{\rm bw} n (4k)^{11}$, where $c_{\rm bw}$ is taken to be $c_{\rm bw} = 4.6\times 10^7$. 

The first interval extends from the local thermalization time to the 8-design time.
Note that $t = c_8 n$, with the constant $c_8:= c_{\rm bw} 8^{11}$, is the circuit depth where we form a unitary 8-design, and the first time step where \autoref{thm:appdesignbound} can be applied for $k=2$, thus giving a stronger suppression than the early time bound. The second interval extends from the 8-design time to the 40-design time $t=c_{40} n$, where the constant $c_{40}:= c_{\rm bw} 40^{11}$. Such a time scale is chosen as a matter of convenience to ensure enough suppression to extend to exponential times.
As such, the third interval extends from the 40-design time to the $d$-design time as given in \autoref{prop:d5}. Finally, the fourth interval extends out to a timescale $e^{c_{\rm rec} d}$ where $c_{\rm rec}<1$ is a constant we will give explicitly.

Let us begin by bounding the probability of a fluctuation in $T_1$.  A union bound tells us that the probability of the union of a number of events, i.e.~the probability that any one of those events occurs, is upper bounded by the sum of the probabilities of those events.
Then considering \autoref{thm:earlytimebound}, the probability of an entropy fluctuation by more than $\delta$ on the time interval $T_1$ is upper bounded by
\begin{align}
T_1:\qquad \Pr\Big( \bigcup_{T_1} \{S(\rho_A(t))\geq \log(d_A) - \delta\} \Big) &\leq \sum_{t = c_{\rm th} \log(d_A)}^{c_8 n} \frac{1}{e^{\delta} - 1} \left(\frac{d_A}{d_B} + d_A \, \left(\frac{2q}{q^2+1}\right)^{2(t - 1)}\right) \\
&\leq \int_{c_{\rm th} \log(d_A)-1}^{c_8 n} \!dt \,  \frac{1}{e^{\delta} - 1} \left(\frac{d_A}{d_B} + d_A \, \left(\frac{4}{5}\right)^{2(t - 1)}\right) \\
&\leq \frac{1}{e^{\delta} - 1} \left(\frac{d_A}{d_B} \big(c_8 n + 1 - c_{\rm th} \log(d_A)\big) + \frac{11}{2}\left(\frac{1}{d_A}\right)^{\frac{2}{5} c_{\rm th}-1 }\right) \\
&\leq \frac{6}{e^{\delta} - 1} \, \left(\frac{1}{d_A}\right)^{\frac{2}{5} c_{\rm th}-1}\,,
\label{eq:thermtime}
\end{align}
where we have used a union bound in the first line, an integral bound and $q \geq 2$ in the second line, and in the last line we have made the mild assumption that $d_B \geq 2 c_8 n d_A^{\frac{2}{5}c_{\rm th}}$.

Next, we bound the probability of an entropy fluctuation by more than $\delta$ on the time interval $T_2$ using \autoref{thm:appdesignbound} for $k=2$.  Again using a union bound and the fact the RQC forms an $8$-design for $t \geq c_8 n$, we find the upper bound
\begin{align}
T_2:\qquad \Pr\Big( \bigcup_{T_2} \{S(\rho_A(t))\geq \log(d_A) - \delta\} \Big) &\leq \sum_{t = c_8 n}^{c_{40} n} \Pr\big(S(\rho_A(t))\geq \log(d_A) - \delta \big)\hspace{2cm}\\
&\leq 2\left(\frac{9\pi^3 d_A^2}{\gamma^2}\frac{2}{d}\right)^2 (c_{40} n - c_8 n)
\label{eq:8design}
\end{align}
where $\gamma:=e^\delta - 1 - \frac{d_A}{d_B}$. Similarly, for $T_3$ we bound the probability of a fluctuation after the 40-design time up to time $t=32n^2 d^5$ using \autoref{thm:appdesignbound} for $k=10$ as
\begin{align}
T_3:\qquad \Pr\Big( \bigcup_{T_3} \{S(\rho_A(t))\geq \log(d_A) - \delta\} \Big) &\leq \sum_{t = c_{40} n}^{32 n^2 d^5} \Pr\big(S(\rho_A(t))\geq \log(d_A) - \delta \big) \hspace{2cm}\\
&\leq 2\left(\frac{9\pi^3 d_A^2}{\gamma^2}\frac{10}{d}\right)^{10} (32n^2d^5 - c_{40} n)\,.
\label{eq:40design}
\end{align}
Now note that by the end of the time interval $T_3$, our RQC is a $d$-design due to \autoref{prop:d5}.  Since \autoref{thm:appdesignbound} provides an upper bound on the probability of entropy fluctuations by more than $\delta$ for a fixed $4k$-design, to bound the behavior for $d$-designs we can minimize over any integer $4k \in [1,d]$. Choosing $k = \lfloor \frac{2d}{e} \frac{\gamma^2}{9\pi^3 d_A^2} \rfloor$, which we note is less than $d/4$ for all values of our parameters, we find
\begin{align}
T_4:\qquad \Pr\Big( \bigcup_{T_4} \{S(\rho_A(t))\geq \log(d_A) - \delta\} \Big) &\leq \sum_{t = 32n^2d^5}^{e^{c_{\rm rec} d}} \Pr\big(S(\rho_A(t))\geq \log(d_A) - \delta \big) \hspace{2cm}\\
&\leq 2 \exp\left( -d \left(\frac{2\gamma^2}{9\pi^3 d_A^2 e} - c_{\rm rec} \right)\right)\,.
\label{eq:ddesign}
\end{align}
Choosing $c_{\rm rec} = \gamma^2/(9 \pi^3 d_A^2 e)$ the result is $e^{-O(d)}$. Putting everything together, the condition that $d_B \geq 2 c_8 n d_A^{\frac{2}{5}c_{\rm th}}$, as well as $c_{\rm th}\geq 3$ and $\delta\geq 2d_A/d_B$, is enough to ensure that Eqs.~\eqref{eq:8design}, \eqref{eq:40design}, and \eqref{eq:ddesign} are less than $\frac{1}{2}\frac{1}{e^{\delta} - 1} d_A^{-\frac{2}{5} c_{\rm th}+1}$ as in Eq.~\eqref{eq:thermtime}. The result is then
\begin{equation}
\Pr\left( N^{\rm ent}_A(\delta)>0 \right)
=\Pr\left( \bigcup\nolimits_{t=c_{\rm th}\log(d_A)}^{e^{c_{\rm rec} d}} \{S(\rho_A(t))\geq \log(d_A) - \delta\} \right) \leq \frac{8}{e^{\delta} - 1} \left(\frac{1}{d_A}\right)^{\frac{2}{5} c_{\rm th} - 1}\,,
\end{equation}
which is the desired bound.
\end{proof}
We gave explicit conditions on the choice of parameters for the theorem to hold. As an example, for a system size of $n=100$ qubits, a subsystem of $n_A=10$ qubits, fluctuation size $\delta=0.01$, and taking $c_{\rm th}=8$, the probability of a single fluctuation from just after the thermalization time to before the recurrence time is $<2\times 10^{-4}$.

We end by noting it follows immediately from the above theorem that the probability of zero fluctuations of the subsystem entropy in the time interval $[c_{\rm th} \log (d_A), e^{c_{\rm rec} d}]$ is lower bounded as
\begin{equation}
    \Pr \Big( N_A^{\rm ent}(\delta) = 0\Big) = \Pr \left(\bigcap\nolimits_{t=c_{\rm th}\log(d_A)}^{e^{c_{\rm rec} d}} \{S(\rho_A(t))\geq \log(d_A) - \delta\} \right) \geq 1-\frac{8}{e^{\delta} - 1} \left(\frac{1}{d_A}\right)^{\frac{2}{5} c_{\rm th} - 1}\,,
\end{equation}
with the same conditions as above. An analogous bound for zero fluctuations in the trace distance similarly holds.

\section{Obstacles to fast thermalization in many-body systems}
\label{app:therm}

In this Appendix, we discuss parts of the literature on thermalization in closed many-body systems, namely those relevant to the discussion in \autoref{sec:numerics}.  For general reviews, see e.g.\ \cite{d2016quantum, gogolin2016equilibration}. 

First we review basic features of Hamiltonian evolution and infinite-time averages, before turning to earlier times.  When the initial state is fine-tuned as the superposition of a small number of eigenstates, subsystems may fluctuate significantly without thermalizing.  Many results about thermalization therefore require that the initial state cannot have large overlap with any eigenstate.  This assumption may be characterized by the ``effective dimension'' \cite{linden2009quantum} of the initial state, defined with respect to the Hamiltonian as 
\begin{equation}
    d_{\textrm{eff}} = \Big(\sum_i c_i^4\Big)^{-1}\,,
\end{equation}
where $c_i$ are the coefficients of the initial state in the energy eigenbasis.  The effective dimension is large when the initial state has small overlap with many eigenstates and large overlap with none.  As elaborated further below, $d_{\textrm{eff}}$ often scales exponentially in total system size.

When the effective dimension is large, infinite-time averages of fluctuations are small.  In \cite{linden2009quantum}, the authors study the reduced state $\rho_A(t)$ of a subsystem $A$ to quantify the fluctuation $F(t)=\norm{\rho_A(t)-\rho_A^{\textrm{avg}}}_1$, where $\rho_A^{\textrm{avg}}$ is the infinite-time average of $\rho_A(t)$.  They prove 
\begin{align} \label{eq:inf_time_ham_mean}
    \langle F(t) \rangle_t = \lim_{T \to \infty} \frac{1}{T} \int_0^T \, dt  \norm{\rho_A(t)-\rho_A^{\textrm{avg}}}_1  \leq d_A d_{\textrm{eff}}^{-\frac{1}{2}},
\end{align}
where $\langle F(t) \rangle_t $ is the infinite time-averaged fluctuation.  Moreover, on infinite time scales, the fraction of the time when $F(t)$ significantly deviates from $\frac{1}{d_{\textrm{eff}}}$ is suppressed as $e^{-d_{\textrm{eff}}}$ (Theorem 4, \cite{linden2009quantum}), i.e.\ schematically
\begin{align} \label{eq:inf_time_ham_conc}
    \textrm{Pr}_t\left(F(t) > d_{\textrm{eff}}^{-\frac{1}{2}} + \epsilon\right) \lesssim e^{-\epsilon^4 d_{\textrm{eff}}}\,.
\end{align}

The strength of these bounds relies on the largeness of $d_{\textrm{eff}}$ for typical conditions.  For generic initial states of chaotic systems, including simple states like product states, one expects $d_{\textrm{eff}}$ is indeed exponential in total system size~\cite{reimann2008foundation, linden2009quantum}; see e.g.\ numerics in \cite{dooley2020enhancing}. Under weak assumptions, the authors of \cite{farrelly2017thermalization} showed the effective dimension $d_{\textrm{eff}}$ is sufficiently large to ensure the fluctuations from equilibrium are small, if not exponentially suppressed.  Alternatively, assuming eigenstates have extensive Renyi-2 entanglement entropy, \cite{wilming2019entanglement} showed $d_{\textrm{eff}}$ is exponential in system size for initial product states; see comments in \cite{rolandi2020extensive}. (For related work addressing unequal-time correlation functions, see \cite{alhambra2020time}.)  In the expected case $d_{\textrm{eff}} \sim e^n$ for generic chaotic systems, $O(1)$ fluctuations must be doubly exponentially suppressed at extremely late timescales.  The suppression then resembles that of Haar random states, or random states chosen from large subspaces \cite{popescu2006foundations}. 

Even in chaotic many-body systems, certain fine-tuned initial states may undergo extremely slow thermalization or periodic revivals within experimentally accessible timescales \cite{schauss2012observation,labuhn2016tunable,choi2019emergent, dooley2020enhancing}, by a mechanism known as many-body scarring.  In these systems, certain initial states have small  $d_{\textrm{eff}}$, related to the presence of eigenstates that violate the eigenstate thermalization hypothesis (ETH) \cite{d2016quantum,srednicki1994chaos,deutsch1991quantum}. Nonetheless, generic chaotic systems are expected to satisfy ``strong'' ETH \cite{kim2014testing}, whereby every eigenstate satisfies ETH, prohibiting short-time revivals and scars \cite{choi2019emergent, alhambra2020revivals}. 
In \autoref{sec:numerics}, we focus on such generic chaotic systems, where we assume initial states have $d_{\textrm{eff}}$ exponential in system size.  Even when making this strong assumptions about thermalization, the statistics of fluctuations at sub-exponential times were not previously well-established. 

The tools of \cite{reimann2008foundation,farrelly2017thermalization, linden2009quantum, linden2010speed,short2012quantum} used to address infinite or exponential timescales appear unsuited for studying earlier times.  An alternative approach studies the relaxation to equilibrium and subsequent fluctuations using random matrix theory techniques \cite{brandao2012convergence,vinayak2012sub,cramer2012thermalization,reimann2016typical,reimann2019transportless,ChaosRMT,cotler2019spectral, wilming2017towards,de2018equilibration}.
The time-dependence of late-time fluctuations is hypothesized \cite{reimann2019transportless} to resemble the time-dependence of quantities like the spectral form factor \cite{reimann2019transportless}, which can be computed for random matrices.  However, for most physical Hamiltonians, such quantities can only be computed numerically.

In the numerics in \autoref{sec:numerics}, we focus on a nonintegrable spin chain with known chaotic dynamics.  However, the fast exponential relaxation exhibited in \autoref{fig:Ham_relax} may be criticized as uncharacteristic for physical systems, where relaxation often appears slower \cite{reimann2016typical}.  For inhomogeneous initial conditions, after local equilibration at $O(1)$ times and the vanishing of transients, relaxation is dominated by hydrodynamic transport such as energy diffusion \cite{lux2014hydrodynamic}. (Observables that do not overlap the hydrodynamic quantities are expected to decay faster.) Under diffusion, spatial modes of wavenumber $k$ decay as $e^{-D k^2t}$ for diffusion constant $D$, and the smallest such mode is $\sim n^{-1}$, so exponential decay is seen only after times of $O(n^2)$.  This slow relaxation often appears to dominate even for homogeneous initial conditions \cite{lux2014hydrodynamic, rakovszky2019entanglement, rakovszky2019sub}.  

Another apparent barrier to exponential relaxation is the phenomenon of pre-thermalization, characterized by a separation of relaxation timescales \cite{mori2018thermalization}. The state relaxes quickly to a non-thermal approximate steadystate, ultimately reaching a thermal state on a larger timescale via slow exponential relaxation. 

We therefore emphasize that when analyzing fluctuation statistics, the analogy between RQCs and Hamiltonian evolution must be made at times after both pre-thermalization and hydrodynamic relaxation timescales.

\section{Relaxation and fluctuations for charge-conserving RQCs} \label{app:ccRQC}

To study the interplay of hydrodynamics and chaos, one desires a model with both the simplicity of RQCs and the hydrodynamic behavior of Hamiltonian evolution. Following a fruitful recent approach due to \cite{rakovszky2018diffusive, khemani2018operator}, we can study random quantum circuits with conserved charges.  In the simplest example, we consider the operator $\sum_{i=1}^n \sigma^z_i$ as the total $U(1)$ charge on all $n$ qubits.  Then instead of drawing two-qubit gates from the Haar ensemble on two qubits, we require the gate to commute with the charge operator $\sigma^z_1 +\sigma^z_2$ on the pair.  All two-qubit operators which commute with the charge $\sigma^z_1 +\sigma^z_2$ can be simultaneously block-diagonalized; one defines the circuit gate using the direct sum of Haar ensembles on each block.  The resulting circuit commutes with total charge, giving rise to the local conserved quantity $\sigma^z_i$.  

We consider brickwork RQCs in 1D, where the gates are drawn according to the above ensemble.  Such models are suggested to characterize universal behaviors of chaotic Hamiltonian systems, including hydrodynamic relaxation, where charge conservation plays the role of energy conservation.

Armed with this toy model, we can more cleanly ask which aspects of the relaxation quantified in \autoref{thm:earlytimebound} persist in the presence of a local conserved quantity.   As discussed in \autorefapp{app:therm}, when the initial state is highly inhomogeneous (e.g.\ all the charge resides on one half of the spin chain), the slow diffusive transport implies global equilibration requires time at least $O(n^2)$.  However, should we still expect that afterward, subsystem fluctuations quickly decrease until they have typical size $1/d$?  Moreover, what if the initial state is already homogeneous?  Without providing definitive answers, we probe these questions with the figures below.  Related numerical investigations are reported in \cite{rakovszky2019sub, rakovszky2019entanglement}, and related analysis of correlation functions appears in \cite{delacretaz2020heavy}.

\autoref{fig:ccRQC_homog} depicts fluctuation statistics when the system is initiated in a homogeneous product state of fixed charge at half-filling. (The initial state has a charge on every even site.) Meanwhile, \autoref{fig:ccRQC_inhomog} depicts analogous data, but where all the charge is initially gathered on the left half of the system.  In the inhomogeneous case, \autoref{fig:ccRQC_inhomog} demonstrates that larger systems exhibit larger fluctuations at early times but smaller fluctuations at later times.  In the homogeneous case, the behavior may appear similar to that of RQCs without charge conservation.  However, as discussed in \cite{rakovszky2019sub}, the decay of fluctuations until $O(n^2)$ times is likely characterized by power law decay.  After $O(n^2)$ time, the classical diffusion equation in 1D would suggest that fluctuations decay as $e^{-D t/n^2}$ for diffusion constant $D$, dominated by the slow exponential decay of the longest-wavelength mode. If this decay continued until time $O(n^3)$, the fluctuations would then have typical size $\sim e^{-n} \sim \frac{1}{d}$.

In summary, for both homogeneous and inhomogeneous initial conditions, we expect the fluctuations already become suppressed like $\frac{1}{d}$ after $\textrm{poly}(n)$ times. This expectation is consistent with our numerics. However, more rigorous analysis or more extensive numerics are required for a definitive conclusion.

\begin{figure}[hb]
\centering
\begin{tikzpicture}[scale=1,baseline=0mm]
\node at (0,0) {\includegraphics[width=0.5\linewidth]{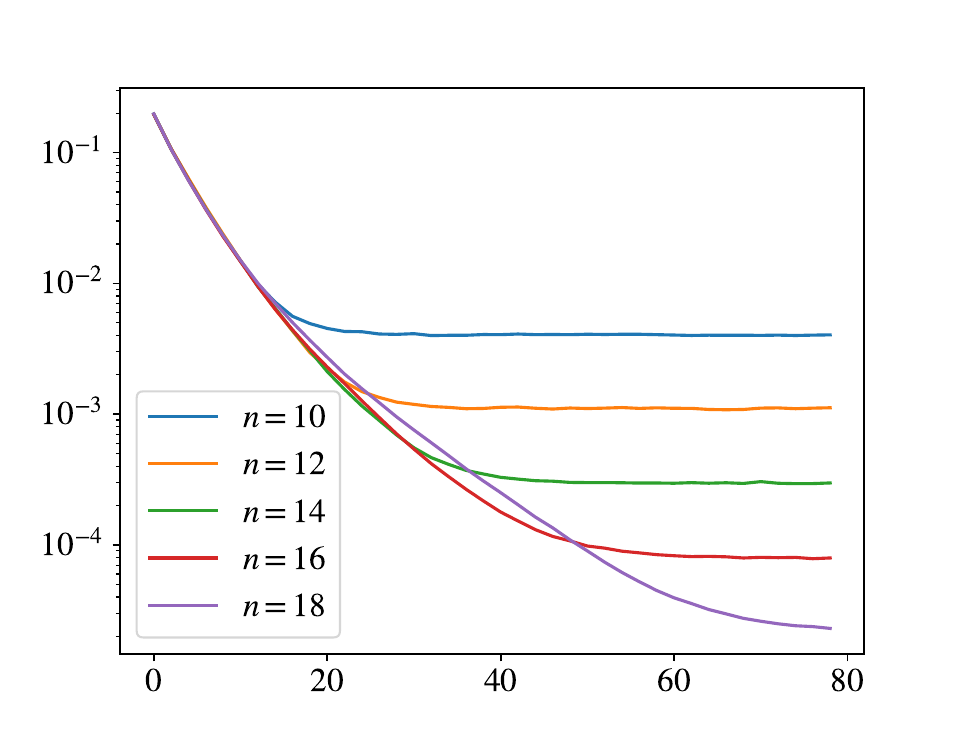}};
\node at (0.6,-3.6) {Time (circuit depth)};
\node[rotate=90] at (-4.7,0.2) {RMS entropy over trials};
\node at (0.4,3.6) {{\footnotesize\bf Charge-conserving RQCs, homogeneous initial charge density}};
\end{tikzpicture}
\caption{We simulate 3000 trials of a brickwork RQC with $U(1)$ conserved charge, on periodic chains of varying length $n$. For a fixed subsystem consisting of one qubit, at each time, we calculate the RMS deviation of the entropy from its trial average, shown on the $y$-axis.  In each trial, the initial state is the product state of fixed charge at half-filling, with charges present at every second site.  Thus the initial condition is spatially homogeneous over the scale of the system size.
}
\label{fig:ccRQC_homog}
\end{figure}
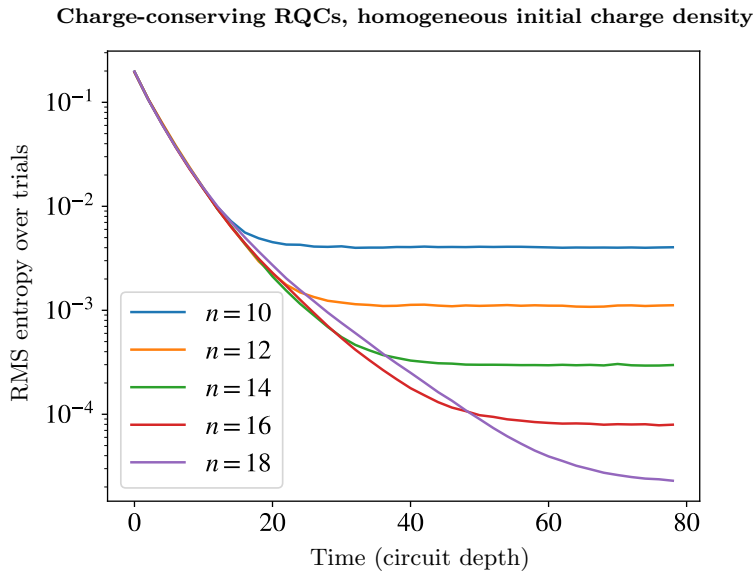

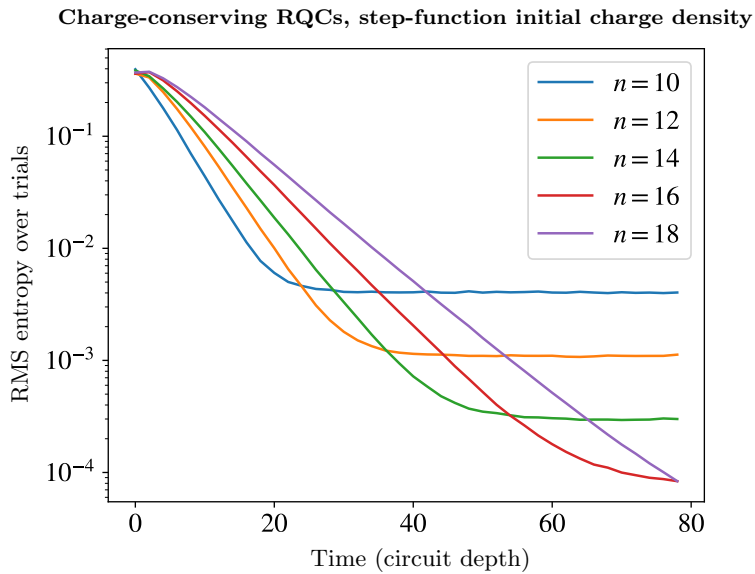
\begin{figure}[hb]
\centering
\begin{tikzpicture}[scale=1,baseline=0mm]
\node at (0,0) {\includegraphics[width=0.5\linewidth]{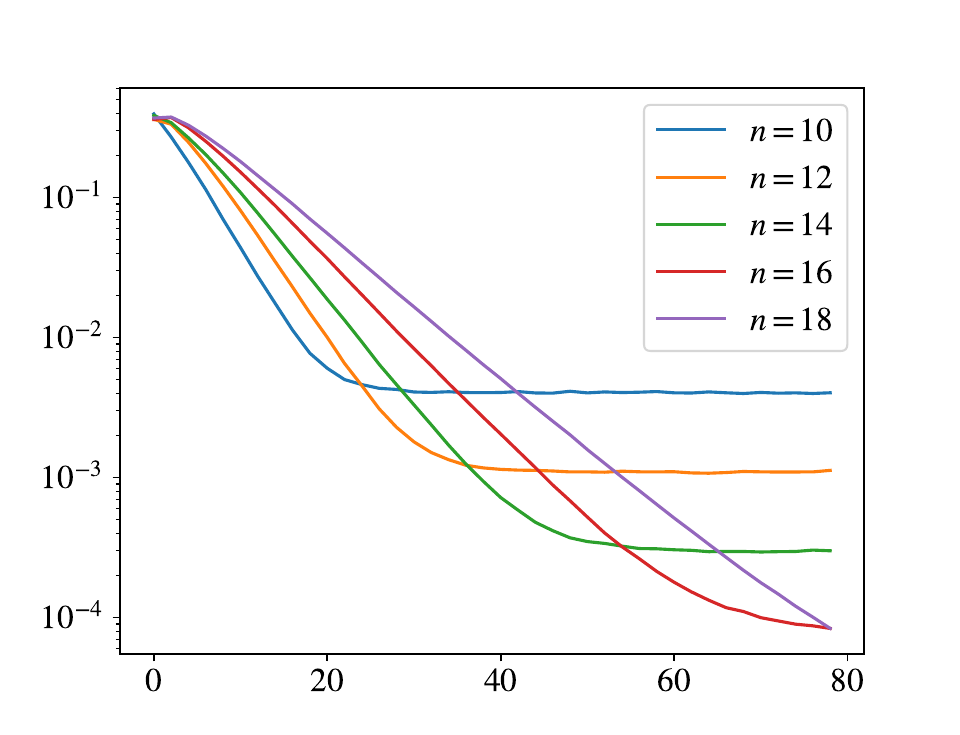}};
\node at (0.6,-3.6) {Time (circuit depth)};
\node[rotate=90] at (-4.7,0.2) {RMS entropy over trials};
\node at (0.4,3.6) {{\footnotesize\bf Charge-conserving RQCs, step-function initial charge density}};
\end{tikzpicture}
\caption{We plot the same kinds of curves as in \autoref{fig:ccRQC_homog}, where the initial state for each trial is again a product state of fixed charge at half filling but all of the charge resides on one half of the system. }
\label{fig:ccRQC_inhomog}
\end{figure}

\bibliographystyle{utphys}
\bibliography{entflucs}

\end{document}